\documentclass[a4paper,USenglish,cleveref,autoref,thm-restate]{lipics-v2021}
\nolinenumbers

\title{Saturation Problems for Families of Automata}

\author{León Bohn}{RWTH Aachen University, Germany}{bohn@lics.rwth-aachen.de}{https://orcid.org/0000-0003-0881-3199}{Supported by DFG grant LO 1174/7-1}
\author{Yong Li}{Key Laboratory of System Software (Chinese Academy of Sciences) and State Key Laboratory of Computer Science, Institute of Software, Chinese Academy of Sciences, PRC}{liyong@ios.ac.cn}{https://orcid.org/0000-0002-7301-9234}{Supported by National Natural Science Foundation of China (Grant No. 62102407)}
\author{Christof Löding}{RWTH Aachen University, Germany}{loeding@automata.rwth-aachen.de}{https://orcid.org/0000-0002-1529-2806}{}
\author{Sven Schewe}{University of Liverpool, UK}{sven.schewe@liverpool.ac.uk}{https://orcid.org/0000-0002-9093-9518}{Supported by the EPSRC through projects EP/X03688X/1 and EP/X042596/1}

\authorrunning{L. Bohn, Y. Li, C. Löding, S. Schewe}
\Copyright{León Bohn, Yong Li, Christof Löding, and Sven Schewe}

\ccsdesc[100]{Theory of computation~Automata over infinite objects}
\keywords{Families of Automata, automata learning, FDFAs}
\category{Track B: Automata, Logic, Semantics, and Theory of Programming}
\EventEditors{Keren Censor-Hillel, Fabrizio Grandoni, Joel Ouaknine, and Gabriele Puppis}
\EventNoEds{4}
\EventLongTitle{52nd International Colloquium on Automata, Languages, and Programming (ICALP 2025)}
\EventShortTitle{ICALP 2025}
\EventAcronym{ICALP}
\EventYear{2025}
\EventDate{July 8--11, 2025}
\EventLocation{Aarhus, Denmark}
\EventLogo{}
\SeriesVolume{334}
\ArticleNo{151}

\usepackage[table,dvipsnames]{xcolor}
\usepackage{xspace}

\usepackage{amssymb,amsmath,amsthm}
\usepackage{mathtools,thmtools} 
\usepackage{tikz}
\usetikzlibrary{shapes.geometric, decorations, calc,arrows,decorations.pathreplacing}
\usetikzlibrary {arrows.meta}
\usetikzlibrary{decorations.pathmorphing}
\usetikzlibrary{automata,positioning}

\tikzset{
    box/.style={
        draw,
        rectangle,
        minimum width=2cm,
        minimum height=1cm,
        text centered,
        node distance=3cm,
    },
    softbox/.style={
        draw,
        rectangle,
        rounded corners=8pt,
        minimum width=2cm,
        minimum height=1cm,
        text centered,
        node distance=3cm,
        thick
    },
    rounded/.style={
        draw,
        ellipse,
        minimum width=2cm,
        minimum height=1cm,
        text centered,
        node distance=3cm,
    },
    squiggly/.style={
        decorate,
        decoration={
            snake,
            amplitude=.4mm,
            segment length=2mm,
            post length=1mm,
            pre length=1mm
        },
        thick
    }
}

\tikzset{automaton/.append style={
			shorten >=1pt,
			auto,
			on grid,
			node distance=10mm,
			initial text=,
			every edge/.append style={
					every node/.append style={
							font=\scriptsize
						}
				},
			every state/.append style={
					inner sep=1pt,
					minimum size=3mm,
                    font=\small,
                    draw=none,
				},
			squig/.append style={decorate,decoration={snake, amplitude=0.2mm, segment length=1.5mm}}
		},
        lblsmall/.append style={
            every node/.append style={
                font=\small
            },
			every edge/.append style={
					every node/.append style={
							font=\small
						}
				},
        },
        lbltiny/.append style={
            every node/.append style={
                font=\scriptsize
            },
			every edge/.append style={
					every node/.append style={
							font=\scriptsize
						}
				},
        }
}

\usepackage{mathcommand}
\usepackage[notion, quotation,hyperref,electronic]{knowledge}

\definecolor{blueish}{RGB}{0,84,159}
\definecolor{petrolish}{RGB}{0,97,101}
\definecolor{violett}{RGB}{97,33,88}

\definecolor{primary}{rgb}{0.0,0.25,0.42}

\definecolor{secondary}{rgb}{0.44, 0.11, 0.11}

\definecolor{tertiary}{rgb}{0.32,0.18,0.5}

\definecolor{quartary}{rgb}{0.45,0.42,0.75}

\newcommand{\mc}[1]{\ensuremath{\mathcal{#1}}}

\newcommand{\reprs}[1][\Sigma]{\ensuremath{{#1}^* \times {#1}^+}}

\newcommand{\A}{\mc{A}}
\newcommand{\T}{\mc{T}}
\newcommand{\F}{\mc{F}}

\newcommand{\N}{\mc{N}}
\newcommand{\D}{\mc{D}}
\newcommand{\C}{\mc{C}}
\renewcommand{\L}{\mc{L}}
\newcommand{\B}{\mc{B}}
\newcommand{\W}{\mc{W}}
\renewcommand{\P}{\mc{P}}

\newcommand{\eps}{\ensuremath{\varepsilon}}

\newcommand{\op}[1]{\ensuremath{\operatorname{\mathsf{#1}}}}

\newcommand{\PSPACE}{\textsf{PSPACE}\xspace}
\newcommand{\PTIME}{\textsf{PTIME}\xspace}

\newcommand{\moracle}{\textit{mem}}
\newcommand{\eoracle}{\textit{eq}}

\IfKnowledgePaperModeTF{
}{
  \knowledgestyle{intro notion}{color={secondary}, emphasize}
  \knowledgestyle{notion}{color={primary}}
  \hypersetup{
    colorlinks=true,
    breaklinks=true,
    linkcolor={tertiary}, 
    citecolor={quartary}, 
  }
}
\knowledgeconfigureenvironment {proof}{}

\knowledge{notion}
| ultimately periodic
| ultimately periodic word
\knowledgenewrobustcmd{\UP}[1]{\mathop{\cmdkl{\mathsf{UP}}}(#1)}

\knowledge{notion}
| loop
| loops

\knowledge{notion}
| transition system
| TS
\knowledgenewrobustcmd{\states}[1]{\mathop{\cmdkl{Q}}(#1)}
\knowledge{notion}
| size@TS

\knowledge{notion}
| strongly connected component
| strongly connected components

\knowledge{notion}
| deterministic finite automaton
| DFA
| DFAs
\knowledgenewrobustcmd{\DFA}[1]{\mathop{\cmdkl{\mathsf{DFA}}}(#1)}
\knowledgenewrobustcmd{\LangDFA}[1]{\mathop{\cmdkl{L_*}}(#1)}
\knowledge{notion}
| nondeterministic finite automaton
| NFA

\knowledge{notion}
| Moore machine

\knowledge{notion}
| regular
| regular $\omega$-language
| regular $\omega$-languages

\knowledge{notion}
| UP-regular

\knowledge{notion}
| UP-equivalent

\knowledge{notion}
| deterministic Büchi automaton
| DBA

\knowledgenewrobustcmd{\LangDBA}[1]{\mathop{L_\omega}(#1)}

\knowledge{notion}
| weak

\knowledge{notion}
| nondeterministic Büchi automaton
| NBA

\let\root\undefined
\knowledgenewrobustcmd{\wordroot}[1]{\cmdkl{\sqrt{#1}}}
\knowledge{synonym}
| root

\knowledgenewrobustcmd{\Pow}[1]{\cmdkl{\operatorname{\omega-\mathsf{Pow}}}(#1)}

\knowledge{notion}
| prefix-independent
| $\sim_L$ is trivial

\knowledge{notion}
| family of deterministic finite automata
| FDFA
| FDFAs

\knowledge{notion}
| leading transition system
| leading TS

\knowledge{notion}
| progress DFA
| progress DFAs

\knowledge{notion}
| syntactic FDFA

\knowledge{notion}
| family of right congruences
| FORC

\knowledge{notion}
| refines

\knowledge{notion}
| language recognizing FORC
| LFORC

\knowledge{notion}
| normalized (with regard to $\F$)
| loops on $u$
| loops on
| loop on
| loops on some class
| normalized

\knowledgenewrobustcmd{\ReprFDFA}[1]{\mathop{\cmdkl{{\mathsf{L}}_{\mathsf{rep}}}}(#1)}
\knowledgenewrobustcmd{\NormFDFA}[1]{\mathop{\cmdkl{{\mathsf{L}}_{\mathsf{norm}}}}(#1)}
\knowledge{notion}
| language@FDFA
| accepted@FDFA
| accepts@FDFA
\knowledgenewrobustcmd{\LangFDFA}[1]{\mathop{\cmdkl{L_\omega}}(#1)}

\knowledgenewrobustcmd{\ReprFDWA}[1]{\mathop{\cmdkl{{\mathsf{L}}_{\mathsf{rep}}}}(#1)}
\knowledgenewrobustcmd{\NormFDWA}[1]{\mathop{\cmdkl{{\mathsf{L}}_{\mathsf{norm}}}}(#1)}
\knowledge{notion}
| language@FDWA
| accepts@FDWA
\knowledgenewrobustcmd{\LangFDWA}[1]{\mathop{\cmdkl{L_\omega}}(#1)}

\knowledge{notion}
| family
| families

\knowledge{notion}
| prefix-independent@family
\knowledge{notion}
| size@family
\knowledge{notion}
| language@family

\knowledgenewmathcommandPIE{\redequiv}{\mathrel{\cmdkl{\equiv#2#3}}}

\knowledge{notion}
| saturated for $\X$

\knowledge{notion}
| saturated
| saturation
| Saturation
| saturated FDFA
| saturated FDFAs
| saturated FDWA
| Saturated FDWAs
| saturated FDWAs

\knowledge{notion}
| almost saturated
| almost saturation
| almost saturated FDFA
| almost saturated FDFAs

\knowledge{notion}
| full saturation
| fully saturated
| fully saturated FDFA
| fully saturated FDFAs

\knowledge{notion}
| reference set

\knowledge{notion}
| loopshift-stable
| loopshift-stability

\knowledge{notion}
| loopshift-stable@for
| loopshift-stability@for
| loopshift-stable for $\X$
| loopshift-stable for

\knowledge{notion}
| power-stable
| power-stability

\knowledge{notion}
| power-stable@for
| power-stability@for
| power-stable for $\X$
| power-stable for

\knowledge{notion}
| upward power-stable

\knowledge{notion}
| family of deterministic weak automata
| FDWA
| FDWAs

\knowledge{notion}
| progress automaton

\knowledge{notion}
| Transition profiles
| transition profile

\knowledge{notion}
| active learner

\knowledge{notion}
| passive learner

\knowledge{notion}
| membership oracle

\knowledge{notion}
| equivalence oracle

\knowledge{notion}
| target language

\newcommand{\ifappendix}[1]{\ifappendixelse{#1}{}}

\newcommand{\ifappendixelse}[2]{\ifthenelse{\isundefined{\shouldhaveappendix}}{#2}{#1}}
\newcommand*{\shouldhaveappendix}{}

\newcommand{\iftodoselse}[2]{\ifthenelse{\isundefined{\showtodos}}{#2}{#1}}

\iftodoselse{%
\usepackage{todonotes}
\newcommand{\sven}[1]{\todo[inline,color=teal!10,caption={Sven}]{\textbf{Sven:}
#1}}
\newcommand{\ly}[1]{\todo[inline,color=orange!10,caption={LY}]{\textbf{LY:} #1}}
\newcommand{\leon}[1]{\todo[color=green!10,caption={Leon}]{\textbf{Leon:} #1}}
\newcommand{\leonin}[1]{\todo[inline,color=green!10,caption={Leon}]{\textbf{Leon:} #1}}
\newcommand{\chrin}[1]{\todo[inline,color=Orchid!30,caption={Christof}]{\textbf{Christof:} #1}}
\newcommand{\chr}[1]{\todo[color=Orchid!30,caption={Christof}]{\textbf{Christof:} #1}}
}{%
\newcommand{\sven}[1]{}
\newcommand{\ly}[1]{}
\newcommand{\leon}[1]{}
\newcommand{\leonin}[1]{}
\newcommand{\chr}[1]{}
\newcommand{\chrin}[1]{}
}

\begin{document}

\maketitle

\begin{abstract}
Families of deterministic finite automata (FDFA) represent regular $\omega$-languages through their ultimately periodic words (UP-words).
An FDFA accepts pairs of words, where the first component corresponds to a prefix of the UP-word, and the second component represents a period of that UP-word.
An FDFA is termed saturated if, for each UP-word, either all or none of the pairs representing that UP-word are accepted.
We demonstrate that determining whether a given FDFA is saturated can be accomplished in polynomial time, thus improving the known \PSPACE upper bound by an exponential.
We illustrate the application of this result by presenting the first polynomial learning algorithms for representations of the class of all regular $\omega$-languages.
Furthermore, we establish that deciding a weaker property, referred to as almost saturation, is \PSPACE-complete.
Since FDFAs do not necessarily define regular $\omega$-languages when they are not saturated, we also address the regularity problem and show that it is \PSPACE-complete.
Finally, we explore a variant of FDFAs called families of deterministic weak automata (FDWA), where the semantics for the periodic part of the UP-word considers $\omega$-words instead of finite words. 
We demonstrate that saturation for FDWAs is also decidable in polynomial time, that FDWAs always define regular $\omega$-languages, and we compare the succinctness of these different models.
\end{abstract}%
%
\section{Introduction}
\label{section:introduction}
Regular $\omega$-languages (languages of infinite words) are a useful tool for developing decision procedures in logic that have applications in model checking and synthesis~\cite{BaierK2008,Thomas09}. There are many different formalisms for representing regular $\omega$-languages, like regular expressions, automata, semigroups, and logic~\cite{Thomas90,PerrinP04}. In this paper, we study \emph{families of automata} that represent regular $\omega$-languages in terms of the ultimately periodic words that they contain.  This is based on the fact that a regular $\omega$-language is uniquely determined by the ultimately periodic words that it contains, which follows from results by Büchi \cite{Buchi62} on the closure of regular $\omega$-languages under Boolean operations (see also \cite[Fact~1]{CalbrixNP93})).  Ultimately periodic words are of the form $uv^\omega$ for finite words $u,v$ (where $v$ is non-empty). For a regular $\omega$-language $L$, \cite{CalbrixNP93} considers the language $L_\$$ of all finite words of the form $u\$ v$ with $uv^\omega \in L$. They show that $L_\$$ is regular, and from a DFA for $L_\$$ one can construct in polynomial time a nondeterministic B\"uchi automaton for $L$.
A similar approach was used independently by Klarlund in \cite{Klarlund94} who introduces the concept of \emph{families of deterministic finite automata} (FDFA) for representing $\omega$-languages, based on the notion of family of right congruences (FORC) introduced by Maler and Staiger \cite{MalerS97}. Instead of using a single DFA, an FDFA consists of one so-called leading transition system, and so-called progress DFAs, one for each state of the leading transition system.
A pair $(u,v)$ of finite words (with $v$ non-empty) is accepted if $v$ is accepted by the progress DFA of the state reached by $u$ in the leading transition system. As opposed to the $L_\$$ representation from \cite{CalbrixNP93}, the FDFA model of \cite{Klarlund94,MalerS97} only considers pairs $(u,v)$ such that $v$ loops on the state of the leading transition system that is reached by $u$, referred to as \emph{normalized pairs}. So an FDFA for $L$ is an FDFA that accepts all normalized pairs $(u,v)$ with $uv^\omega \in L$. 

Because there exist many learning algorithms for DFAs, these kinds of representations based on DFAs have received attention in the area of learning regular $\omega$-languages
\cite{FarzanCCTW08,AngluinF16,LiCZL21,BohnL22,BohnL24,LiST24}. 
One problem is that methods for DFA learning might come up with automata that treat different pairs that define the same ultimately periodic word differently. So it might happen that a pair $(u_1,v_1)$ is accepted while $(u_2,v_2)$ is rejected, although $u_1v_1^\omega = u_2v_2^\omega$. In this case it is not clear anymore which $\omega$-language is accepted. One can use the existential (nondeterministic) semantics and say that $uv^\omega$ is accepted if some pair representing $uv^\omega$ is accepted. But this does not necessarily define a regular $\omega$-language (see \cite[Example 2]{LiCZL21}).
An FDFA is called \emph{saturated} if it treats all normalized pairs representing the same ultimately periodic word in the same way (accepts all or rejects all). Additionally, we call an FDFA \emph{fully saturated} if it treats all pairs representing the same ultimately periodic words in the same way (not only the normalized ones). The $u\$v$ representation corresponds to fully saturated FDFAs because a DFA for the language $L_\$$ can easily be turned into a fully saturated FDFA for $L$ by taking the part before $\$$ as leading transition system, and the part after the $\$$ for defining the progress DFAs.

It has been shown in \cite{AngluinBF18} that saturated FDFAs possess many good properties as formalism for defining $\omega$-languages.\footnote{The semantics of FDFAs in \cite{AngluinBF18} and \cite{Klarlund94} are defined differently, however, saturated FDFAs define the same $\omega$-languages in both models. We follow the definitions from \cite{AngluinBF18}.} However, up to now the best known upper bound on the complexity for deciding if a given FDFA is (fully) saturated is \PSPACE \cite{AngluinBF18}. In this paper, we settle several questions concerning the complexity of saturation problems. 
As our first contribution we provide polynomial time algorithms for deciding saturation and full saturation of FDFAs, solving a question that was first posed more than three decades ago in \cite{CalbrixNP93} (see related work for existing work on the saturation problem). 
We also consider the property of almost saturation, which is weaker than saturation and thus allows for smaller FDFAs (\Cref{lemma:lowerbound:combinedToSatFDFA}), while ensuring that almost saturated FDFAs define regular $\omega$-languages \cite{LiST23}. We show that this comes at a price, because deciding almost saturation is \PSPACE-complete. Furthermore, we show that it is \PSPACE-complete to decide whether a given FDFA defines a regular $\omega$-language, solving a question that was raised in \cite{CalbrixN95}.

An overview of the models that appear in our paper and of our results is given in Figure~\ref{fig:overview}. The different classes of FDFAs are shown on the left-hand side. The entries \PTIME/\PSPACE-c.\ indicate the complexity of deciding whether a given FDFA belongs to the respective class. As an application of the polynomial saturation algorithms, we provide the \emph{first} polynomial learning algorithms for classes representing all regular $\omega$-languages: a polynomial time active learner for fully saturated FDFAs, and a passive learner for syntactic FDFAs that is polynomial in time and data (see related work for existing work on learning $\omega$-languages). The syntactic FDFA for $L$ is the saturated FDFA with the least possible leading transition system (called $\mathfrak{L}^L$ in \cite{Klarlund94}). 
The size of the models that our algorithms can learn is, in general, incomparable to the size of deterministic parity automata (DPAs) (\cite[Theorem~5.12]{AngluinBF18} shows an exponential lower bound from saturated FDFAs to DPAs and the language family used for that is accepted also by small fully saturated FDFAs; an exponential lower bound from DPAs to syntactic FDFAs follows from the language family $L_n$ used in  \Cref{lemma:expgapsatFDFAsynFDFAfullysatFDFA}, which is easily seen to be accepted by small DPAs).
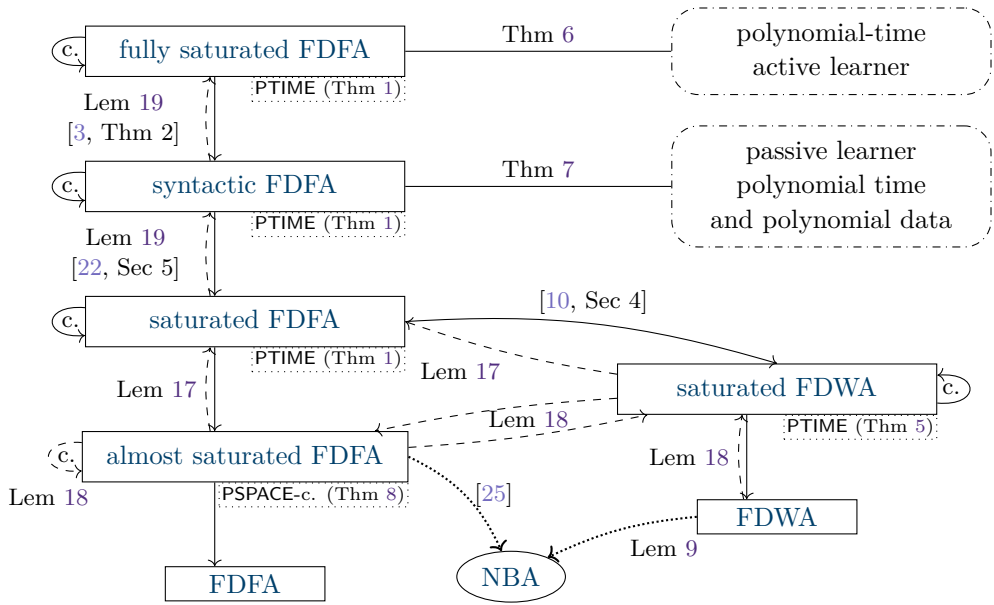
\begin{figure}
    \begin{center}
    \begin{tikzpicture}[%
    baseline={(NBA)},%
    box/.style={%
        draw,%
        rectangle,%
        minimum width=0.3\textwidth,%
        text centered,%
        node distance=3cm,%
        inner sep=3pt,%
    },%
    softbox/.style={%
        draw,%
        rectangle,%
        dash dot,%
        rounded corners=10pt,%
        minimum width=0.4\textwidth,%
        minimum height=1cm,%
        text centered,%
        node distance=3cm,%
        minimum width=0.3\textwidth,%
    },%
    rounded/.style={%
        draw,%
        ellipse,
        text centered,%
        node distance=3cm,%
    }%
    ]
        \node[box] (fullysatFDFA) {\begin{tabular}{c}"fully saturated FDFA"\end{tabular}};
        \node[box, below=0.05\textheight of fullysatFDFA] (synFDFA) {\begin{tabular}{c}"syntactic FDFA"\end{tabular}};
        \node[box, below=0.05\textheight of synFDFA] (satFDFA) {\begin{tabular}{c}"saturated FDFA"\end{tabular}};
        \node[box, below=0.05\textheight of satFDFA] (almostsatFDFA) {\begin{tabular}{c}"almost saturated FDFA"\end{tabular}};
        \node[box, minimum width=0.15\textwidth, below=0.05\textheight of almostsatFDFA] (FDFA) {"FDFA"};

        \node[softbox,right=0.25\textwidth of fullysatFDFA] (ptimeAL) {\begin{tabular}{c}polynomial-time\\active learner\end{tabular}};
        \node[softbox] (passiveLearner) at (synFDFA -| ptimeAL) {\begin{tabular}{c}passive learner\\polynomial time\\and polynomial data\end{tabular}};
        
        \node[box] (satFDWA) at ($(satFDFA)!0.5!(almostsatFDFA) + (0.5\textwidth,0)$) {\begin{tabular}{c}"saturated FDWA"\end{tabular}};
        \node[box, minimum width=0.15\textwidth] (FDWA) at ($(FDFA -| satFDWA) + (0,.04\textheight)$) {"FDWA"};

        \node[rounded] (NBA) at ($(FDFA)!0.5!(FDWA) + (0,-0.025\textwidth)$) {"NBA"};       

        \tikzset{lblsmall}

        \draw[-, double]
        (fullysatFDFA) edge node[anchor=south]{Thm~\ref{theorem:activelearner}} (ptimeAL);
        \draw[-, double]
        (synFDFA) edge node[anchor=south]{Thm~\ref{theorem:passivelearner}} (passiveLearner);

        \draw[densely dotted, thick]
        (FDWA.west) edge[->, bend right=10] node[anchor=north west] {\small Lem~\ref{lemma:FDWAregular}} (NBA)
        (almostsatFDFA.east) edge[->, bend left=20] node[anchor=west] {\small \cite{LiST23}}  (NBA)
        ;

        \draw[->]
        (fullysatFDFA) edge[transform canvas={xshift=-4mm}] (synFDFA)
        (synFDFA) edge[transform canvas={xshift=-4mm}] (satFDFA)
        (satFDFA) edge[transform canvas={xshift=-4mm}] (almostsatFDFA)
        (almostsatFDFA) edge[transform canvas={xshift=-4mm}] (FDFA)
        (satFDWA) edge[transform canvas={xshift=-4mm}]  (FDWA);

        \draw[->, dashed, lblsmall]
        (synFDFA) edge[bend left=10,transform canvas={xshift=-4mm}] node[anchor=east] {\begin{tabular}{c}Lem~\ref{lemma:expgapsatFDFAsynFDFAfullysatFDFA}\\\cite[Thm~2]{AngluinF16} \end{tabular}} (fullysatFDFA)
        (satFDFA) edge[bend left=10,transform canvas={xshift=-4mm}] node[anchor=east] {\begin{tabular}{c}Lem~\ref{lemma:expgapsatFDFAsynFDFAfullysatFDFA}\\\cite[Sec~5]{Klarlund94}\end{tabular}} (synFDFA)
        (almostsatFDFA) edge[bend left=10,transform canvas={xshift=-4mm}] node[anchor=east] {Lem~\ref{lemma:lowerbound:combinedToSatFDFA}} (satFDFA)
        (satFDWA) edge[bend left=5] node[anchor=north east] {Lem~\ref{lemma:lowerbound:combinedToSatFDFA}} (satFDFA.east)
        (almostsatFDFA) edge[bend right=4] node[anchor=south] {Lem~\ref{lemma:incomparableSatFDWAalsatFDFA}} (satFDWA)
        (satFDWA) edge[bend right=4] (almostsatFDFA)
        (FDWA) edge[bend left=10,transform canvas={xshift=-4mm}] node [anchor=east] {Lem~\ref{lemma:incomparableSatFDWAalsatFDFA}} (satFDWA);

        \draw[lblsmall]
        (satFDFA.east) edge[->, bend left=10] node[anchor=south] {\small \cite[Sec~4]{BohnL24}} (satFDWA.north);

        \draw[->]
        (fullysatFDFA) to [out=175, in=185, looseness=4] node[anchor=west] {c.} (fullysatFDFA);
        \draw[->]
        (synFDFA) to [out=175, in=185, looseness=4] node[anchor=west] {c.}  (synFDFA);
        \draw[->]        
        (satFDFA) to [out=175, in=185, looseness=4] node[anchor=west] {c.}  (satFDFA);
        \draw[->]
        (satFDWA) to [out=355, in=5, looseness=4] node[anchor=east] {c.} (satFDWA);
        \draw[->, dashed]
        (almostsatFDFA) to [out=175, in=185, looseness=4] node[anchor=west] {c.} node[anchor=north,overlay,yshift=-3mm] {\small Lem~\ref{lemma:incomparableSatFDWAalsatFDFA}} (almostsatFDFA);

        \node[inner sep=1pt, draw=black, dotted, anchor=north east] (DECfullysatFDFA) at (fullysatFDFA.south east) {\scriptsize \PTIME (Thm~\ref{theorem:ptimesaturationFDFA})};
        \node[inner sep=1pt, draw=black, dotted, anchor=north east] (DECsynFDFA) at (synFDFA.south east) {\scriptsize \PTIME (Thm~\ref{theorem:ptimesaturationFDFA})};
        \node[inner sep=1pt, draw=black, dotted, anchor=north east] (DECsatFDFA) at (satFDFA.south east) {\scriptsize \PTIME (Thm~\ref{theorem:ptimesaturationFDFA})};
        \node[inner sep=1pt, draw=black, dotted, anchor=north east] (DECalmostsatFDFA) at (almostsatFDFA.south east) {\scriptsize \PSPACE-c. (Thm~\ref{theorem:powerclosednessPSPACEcomplete})};
        \node[inner sep=1pt, draw=black, dotted, anchor=north east] (DECsatFDWA) at (satFDWA.south east) {\scriptsize \PTIME (Thm~\ref{theorem:ptimesaturationFDWA})};
    \end{tikzpicture}
    \caption{Overview of the properties and models considered in this paper. 
    Solid arrows indicate translations that are possible without blowup, while dashed ones are used for translations with an exponential lower bound.
    Dotted arrows are constructions yielding automata, and ones labeled with ``c.'' are complementation constructions.
    Solid lines are used to indicate a connection to learning.
    We provide more details on this diagram in the introduction.}
    \label{fig:overview}
    \end{center}
\end{figure}

For completing the picture of the FDFA landscape, we also give examples showing that complementing almost saturated FDFAs can cause an exponential blow-up (while it is trivial for saturated or fully saturated FDFAs), and for exponential size gaps between the different models. Exponential lower bounds are indicated by the dashed arrows in Figure~\ref{fig:overview}, where the loops labeled ``c.''~represent the complement. The solid arrows mean that there is a transformation that does not cause any blow-up. 

Finally, we also consider the model of families of deterministic weak automata (FDWA), which is shown with two entries on the right-hand side of Figure~\ref{fig:overview}. This is an automaton model for families of weak priority mappings (FWPM), which have been introduced in \cite{BohnL24}. Syntactically, an FDWA is an FDFA in which the progress DFAs have the property that each strongly connected component (SCC) of states is either completely accepting or completely rejecting. This corresponds to the restriction made for deterministic weak (B\"uchi) automata (DWA); See, e.g., \cite{Loding01}. And indeed, the progress automata are not interpreted as DFAs accepting finite words, but as deterministic weak automata accepting periodic words: a pair $(u,v)$ is accepted if $v^\omega$ is accepted by the progress DWA of the state reached by $u$ in the leading transition system, which means that $v^\omega$ ends up cycling in an accepting SCC. 
As for FDFAs, the property of saturation for FDWAs requires that either all normalized representations of an ultimately periodic word are accepted or none. (In order to not consider too many models, we restricted ourselves to existing ones and did not additionally introduce fully saturated and syntactic FDWAs.)
We show that saturation can also be decided in polynomial time for FDWAs, which is easier to see than for FDFAs.

Concerning succinctness of FDWAs compared to FDFAs, it follows from results in \cite{BohnL24} that saturated FDFAs can be transformed into saturated FDWAs by only adapting the acceptance mechanism and keeping the transition structure. We show that a translation in the other direction can cause an exponential blow-up, even when going to the more relaxed model of almost saturated FDFAs. Vice versa, we show that translating almost saturated FDFAs to saturated FDWAs also might cause an exponential blow-up.

In summary, our contributions are polynomial saturation algorithms for FDFAs and FDWAs, \PSPACE-completeness results for deciding almost saturation and regularity of FDFAs, and some exponential lower bounds for
transformations between 
different models.

\smallskip 
\noindent\textbf{Organization.} The paper is structured as follows. We continue the introduction with a discussion of related work. 
In \cref{section:preliminaries} we introduce the main definitions. In \cref{sec:saturation} we present the polynomial time saturation algorithms (\cref{section:ptimesaturation}), the learning algorithms (\cref{sec:learning}), and the \PSPACE-completeness of almost saturation (\cref{section:almostsaturation}). The regularity problem is shown to be \PSPACE-complete in \cref{section:regularity}, and the exponential lower bounds for transformations between the models are presented in \cref{section:comparison}. We conclude in \cref{section:conclusion}.

\newcommand{\pow}{\operatorname{\mathsf{pow}}}
\newcommand{\root}{\operatorname{\mathsf{root}}}
\newcommand{\conj}{\operatorname{\mathsf{conj}}}
\newcommand{\per}{\operatorname{\mathsf{per}}}
\smallskip 
\noindent\textbf{Related work on the saturation and regularity problems.}
The problem of saturation was first raised in the conclusion of \cite{CalbrixNP93} for the $L_\$$ representation of ultimately periodic words: ``How can we decide efficiently that a rational language $K \subseteq A^*\$ A^+$ is saturated by $\stackrel{\text{UP}}{\equiv}$?'' Here, two words $u\$ v$ and $u'\$ v'$ are in relation $\stackrel{\text{UP}}{\equiv}$ if $uv^\omega = u'(v')^\omega$.
A variation of the problem appears in \cite{CalbrixN95} in terms of so called period languages:
For an $\omega$-language $L$, a finite word $v$ is called a \emph{period} of $L$ if there is an ultimately periodic word of the form $uv^\omega$ in $L$,
and $\per(L)$ denotes the set of all these periods of $L$.
So $\per(L)$ is obtained from $L_\$$ by removing the prefixes up to (and including) the $\$$.
It is shown that a language $K$ of finite words is a period language iff $K$ is closed under the operations of power, root, and conjugation, where $\pow(K) =  \{v^k \mid v \in K,\ k \geq 1\}$, $\root(K) = \{v \mid \exists k>0:\; v^k \in L\}$, and $\conj(K) = \{v_2v_1 \mid v_1v_2 \in K\}$.
Our notion of "power-stable" corresponds to closure under $\root$ and $\pow$, and our notion of "loopshift-stable" corresponds to closure under conjugation.
Given any $K \subseteq \Sigma^*$, the least period language containing $K$ is $\per(L) = \root(\pow(\conj(K)))$. Then \cite{CalbrixN95} poses three decision problems for a given regular language $K \subseteq \Sigma^*$:
(1) Is $K$ a period language?
(2) Is $\per(K)$ regular?
(3) Is $\pow(K)$ regular?

Problem~(1) is a variation of the saturation problem for FDFAs that focuses solely on periods, independent of their prefixes.
It is observed in \cite{CalbrixN95} that Problem (1) is decidable without considering the complexity. Similarly, Problem~(2) is a variation of the regularity problem that we study for FDFAs, again considering only periods independent of their prefixes. It is left open in \cite{CalbrixN95} whether Problem~(2) is decidable, but a reduction to Problem~(3) is given.
Notably, this reduction is exponential since it constructs a DFA for $\root(\conj(K))$. 
Later, Problem~(3) was shown to be decidable in \cite{Fazekas11} without an analysis of the complexity. Our results in \Cref{section:regularity} give a direct proof that Problem (2) is \PSPACE-complete.

The paper \cite{BealCR96} defines cyclic languages as those that are closed under the operations of power, root, and conjugation, and gives a characterization of regular cyclic languages in terms of so called strongly cyclic languages. So the cyclic languages are precisely the period languages of \cite{CalbrixN95}. Decidability questions are not studied in \cite{BealCR96}.

In \cite{CianciaV19}, lasso automata are considered, which are automata accepting pairs of finite words. The representation basically corresponds to the $u\$ v$ representation of \cite{CalbrixNP93} but without an explicit $\$$ symbol, using different transition relations instead.
It is easy to see that lasso automata and automata for the $u\$ v$ representation are equivalent up to minor rewriting of the automata. 
The property of full saturation corresponds to the notion of bisimulation invariance in  \cite{CianciaV19}: two lassos $(u,v)$ and $(u',v')$ are called bisimilar if $uv^\omega = u'(v')^\omega$, and a lasso automaton is bisimulation invariant if it accepts all pairs from a bisimulation class or none.
Bisimulation invariance is characterized in  \cite{CianciaV19} by the properties \textit{circular} (power-stable in our terminology) and \textit{coherent} (loopshift-stable in our terminology), \cite[Theorem 2]{CianciaV12}, \cite[Section 3.3]{CianciaV19}. An $\Omega$-automaton is defined as a circular and coherent lasso automaton, which corresponds to the property of full saturation for FDFAs. This property of a lasso automaton is shown to be decidable in \cite[Theorem 18]{CianciaV19} without considering the complexity (since the decision procedure builds the monoid from a DFA, it is at least exponential). Another decision procedure is given in \cite[Proposition 12]{ChernevHK24} with a doubly exponential upper bound.

Finally, the saturation problem for FDFAs is explicitly considered in \cite[Theorem 4.7]{AngluinBF18}, where it is shown to be in \PSPACE by giving an upper bound (exponential in the size of the FDFA) on the size of shortest pairs violating saturation if the given FDFA is not saturated.

\smallskip \noindent
\textbf{Related work concerning learning algorithms.}
There are some active and some passive learning algorithms for different representations of regular $\omega$-languages, but all of them are either for subclasses of the regular $\omega$-languages, or they are not polynomial time (for active learners) or not polynomial in the required data (for passive learners), where the complexity for learning a class $\C$ of representations is measured in a smallest representation of the target language from $\C$.
The only known polynomial time active learner for regular $\omega$-languages in the minimally adequate teacher model of \cite{Angluin87} is for deterministic weak Büchi automata \cite{MalerP95}.
There are active learners for nondeterministic Büchi automata (NBA) \cite{FarzanCCTW08,LiCZL21} and for deterministic Büchi and co-Büchi automata \cite{LiST24}, but these are heuristics for the construction of NBA that are not polynomial in the target representations.
The algorithm from \cite{FarzanCCTW08} uses an active learner for DFAs for learning the $u\$v$ representation of \cite{CalbrixNP93} for regular $\omega$-languages. Whenever the DFA learner asks an equivalence query, the DFA is turned into an NBA, and from a counter-example for the NBA a counter example for the DFA is derived. 
Our active learner uses the same principle but turns the DFA for the $u\$v$ representation into an FDFA and checks whether it is fully saturated.

There is a polynomial time active learner for deterministic parity automata \cite{MichaliszynO22}, but this algorithm uses in addition to membership and equivalence queries so-called loop-index queries that provide some information on the target automaton, not just the target language. So this algorithm does not fit into the minimally adquate teacher model from \cite{Angluin87}. The paper \cite{AngluinF16} presents a learning algorithm for FDFAs. But it turns out that this is not a learning algorithm for regular $\omega$-languages: The authors make the assumption that the FDFAs used in equivalence queries define  regular $\omega$-languages (beginning of Section 4 in \cite{AngluinF16}) while they only provide such a semantics for saturated FDFAs. So their algorithm requires an oracle that returns concrete representations of ultimately periodic words as counter-examples for \emph{arbitrary} FDFAs, which is not an oracle for regular $\omega$-languages. Our algorithms solve this problem by first using the polynomial time saturation check, making sure that we only submit saturated FDFAs to the equivalence oracle.

Concerning passive learners, the only polynomial time algorithms that can learn regular $\omega$-languages in the limit from polynomial data are for subclasses that make restrictions on the canonical Myhill/Nerode right-congruence of the potential target languages. The algorithm from \cite{AngluinFS20} can only learn languages for which the minimal automaton has only one state per Myhill/Nerode equivalence class (referred to as IRC languages for ``informative right congruence''). The algorithm from \cite{BohnL21} can also learn languages beyond that class but there is no characterization of the class of languages that can be inferred beyond the IRC languages. The algorithm in \cite{BohnL22} infers deterministic Büchi automata (DBA) in polynomial time but the upper bound on the size of characteristic samples for a DBA-recognizable language $L$ is exponential in a smallest DBA for $L$, in general. The algorithm in \cite{BohnL24} generalizes this to deterministic parity automata (DPA) and can infer a DPA for each regular $\omega$-language, but a polynomial bound for the size of characteristic samples is only given for languages accepted by a  DPA that has at most $k$ states per Myhill/Nerode class for a fixed $k$.

\smallskip \noindent
\textbf{Related work concerning the model of FDWA.}
The paper \cite{FismanGZ24} considers a model of FDFA with so-called duo-normalized acceptance and mentions connections to the model FWPM from \cite{BohnL24} in a few places, without going into details of the connection.
Using results from \cite{BohnL24}, it is not very hard to show that the models can be considered the same in the sense that they can be easily converted into each other without changing the transition structure. We give a proof of this in the full version of the paper.

%
%
\section{Preliminaries}
\label{section:preliminaries}
An alphabet is a finite, non-empty set $\Sigma$ of symbols.
We write $\Sigma^*$ and $\Sigma^\omega$ to denote the set of finite and the set of infinite words over $\Sigma$, respectively. 
For a word $w$, we use $|w|$ to refer to its length, which is the number of symbols if $w$ is finite, and $\infty$ otherwise.
The empty word, $\eps$, is the unique word of length $0$ and we use $\Sigma^+$ for the set of finite, non-empty words, i.e. $\Sigma^+ = \Sigma^* \setminus \{\eps\}$. We assume a linear order on the alphabet and consider the standard length-lexicographic ordering on finite words, which first orders words by length, and by lexicographic order among words of same length.
\AP
An $\omega$-word $w \in \Sigma^\omega$ is called ""ultimately periodic"" if it can be written as $ux^\omega$ with $u \in \Sigma^*$ and $x \in \Sigma^+$.
We write $\intro*\UP{\Sigma}$ for the set of all ultimately periodic words over the alphabet $\Sigma$.
\AP
For a word $x \in \Sigma^*$, we write $\wordroot{x}$ to denote its ""root"", which is the unique minimal word $u$ such that $x^\omega = u^\omega$.\footnote{Usually, $\wordroot{x}$ is defined to be the minimal $u$ with $u^i = x$ for some $i > 0$. A simple application of the Theorem of Fine and Wilf shows that the two definitions coincide: If $u^i = x$ and $v^\omega = x^\omega$, then $v^{|x|} = u^{i|v|} =: y$. Since $y$ has periods $|u|$ and $|v|$, by the theorem of Fine and Wilf, it also has period $gcd(|u|,|v|)$.}
\AP
For set $P \subseteq \Sigma^+$, we write $\intro*\Pow{P}$ to denote the set $\{x^\omega \mid x \in P\}$

\AP
A ""transition system"" (TS) is a tuple $\T = (\Sigma, Q, \delta, \iota)$, consisting of an input alphabet $\Sigma$, a non-empty set of states $Q$, a transition function $\delta : Q \times \Sigma \to Q$, and an initial state $\iota \in Q$.
We define $\delta^* : Q \times \Sigma^* \to Q$ as the natural extension of $\delta$ to finite words, and overload notation treating $\T$ as the function mapping a finite word $u$ to $\delta^*(\iota, u)$.
\AP
We write $\intro*\states{\T}$ for the set of states of a TS-like object $\T$, and assume that $\states{\T}$ is prefix-closed subset of $\Sigma^*$ containing for each state the length-lexicographically minimal word that reaches it.
\AP
The ""size@@TS"" of $\T$, denoted $|\T|$, is the number of states.
\AP
A set $S \subseteq Q$ of states is a ""strongly connected component"" (SCC) if $S$ is non-empty, and $\subseteq$-maximal with respect to the property that for every pair of states $p, q \in S$ there is a non-empty word $x \in \Sigma^+$ such that $\delta^*(p, x) = q$.
\AP
An equivalence relation ${\sim} \subseteq \Sigma^* \times \Sigma^*$ is called right congruence if it preserves right concatenation, meaning $x \sim y$ implies $xz \sim yz$ for all $x,y,z \in \Sigma^*$.
\AP
A "TS" $\T$ can be viewed as a right congruence $\sim_\T$ where $x \sim_\T y$ if $\T(x) = \T(y)$, and a right congruence $\sim$ can be viewed as a TS using the classes of $\sim$ as states and adding for each class $u$ and symbol $a \in \Sigma$ the transition $[u]_\sim \xrightarrow{a} [ua]_\sim$.

\AP
A ""deterministic finite automaton"" (DFA) $\D = (\Sigma, Q, \delta, \iota, F)$ is a "TS" with a set $F \subseteq Q$ of final states.
By replacing $\delta$ with a transition relation $\Delta \subseteq Q \times \Sigma \times Q$, we obtain a ""nondeterministic finite automaton"" (NFA), so every DFA is also an NFA.
\AP
The language accepted by an NFA $\N$, denoted $\intro*\LangDFA{\N}$, consists of all words that can reach a state in $F$.

\AP
A ""deterministic Büchi automaton"" (DBA) $\B$ is syntactically the same as a "DFA".
It accepts an $\omega$-word $w \in \Sigma^\omega$ if the run of $\B$ on $w$ visits a final state from $F$ infinitely often, and we write $\intro*\LangDBA{\B}$ for the language accepted by $\B$.
\AP
If all states within a "strongly connected component" are either accepting, or all of them are rejecting, we call $\B$ ""weak"".
\AP
If we replace the transition function of a "DBA" with a transition relation $\Delta \subseteq Q \times \Sigma \times Q$, we obtain a ""nondeterministic Büchi automaton"" (NBA), which accepts a word $w$ if some run on $w$ visits $F$ infinitely often.
\AP
The language of an NBA is the set of all words it accepts, and say that $L \subseteq \Sigma^\omega$ is a ""regular $\omega$-language"" if it is the language of some NBA.
It already follows from the results of Büchi~\cite{Buchi62} that if two $\omega$-languages agree on the "ultimately periodic" words, then they must be equal.
\AP
We call a set $L\subseteq \UP{\Sigma}$ of "ultimately periodic" words ""UP-regular"" if there exists some "regular $\omega$-language" $L'$ such that $L' \cap \UP{\Sigma} = L$.

\AP
A family of DFAs (""FDFA"") is a pair $\F = (\T, \D)$ where $\T$ is a "transition system" called the ""leading TS"", and $\D$ is an indexed family of automata, so for each state $q \in \states{\T}$, $\D_q$ is a "DFA".
We overload notation writing $\D_u$ for the DFA $\D_{\T(u)}$ and call $\D_u$ a ""progress DFA"".
We always assume that if $x$ and $y$ lead to the same state in some $\D_u$, then also $ux$ and $uy$ lead to the same state in $\T$.
This can always be achieved with a blow-up that is at most quadratic by taking the product of the progress DFA with the "leading TS", and it makes arguments in some proofs (e.g.~\cref{lemma:rotationinvarianceproperty}) simpler.
\AP
An FDFA $\F$ ""accepts@@FDFA"" a pair $(u, x) \in \reprs$ if the DFA $\D_u$ accepts $x$, and we write $\intro*\ReprFDFA\F$ to denote the set of all representations accepted by $\F$.
\AP
A representation $(u,x)$ is called ""normalized (with regard to $\F$)"" if $ux$ and $u$ lead to the same state in $\T$, and write $\intro*\NormFDFA\F$ for the set of "normalized" pairs "accepted@@FDFA" by $\F$.
\AP
A family of deterministic weak automata (""FDWA"") $\W = (\T, \B)$ is syntactically the same as an "FDFA" but it views the progress DFAs as DBAs, and additionally requires each DBA $\B_u$ to be "weak".
$\W$ accepts a pair $(u,x)$ if $x^\omega$ is accepted by $\B_u$ viewed as a "DBA".
We write $\intro*\ReprFDWA{\W}$ for the set of all representations accepted by $\W$ and denote by $\intro*\NormFDWA\W$ the restriction to "normalized" representations.

\AP
As FDFAs and FDWAs are syntactically the same, we say ""family"" $\F = (\T, \D)$ to refer to either an FDFA or an FDWA.
We call $\D_u$ a ""progress automaton"", and write just $\D$ to denote $\D_\eps$ in the case that $\T$ is trivial.
\AP
The ""size@@family"" of $\F$ is the pair $(n,k)$ where $n$ is the "size@@TS" of $\T$ and $k$ is the maximum "size@@TS" of any $\D_u$.
\AP
We call a family $\F$ ""saturated"" if for all "normalized" $(u,x),(v,y)$ with $ux^\omega = vy^\omega$ holds that $(u,x)$ is accepted by $\F$ iff $(v,y)$ is accepted by $\F$.
\AP 
We say that $\F$ is ""fully saturated"" if this equivalence holds for all pairs, not only normalized ones.
For every "saturated FDFA" $\F$, there exists a unique "regular $\omega$-language" $L$ which contains precisely those "ultimately periodic" words that have a "normalized" representation which is accepted by $\F$~\cite[Theorem~5.7]{AngluinBF18}.
In this case we say that $L$ is the \emph{language} accepted by $\F$ and denote it by $\intro*\LangFDFA{\F}$.

The left quotient of the word $u$ with regard to a regular language $L$ of finite or infinite words, denoted $u^{-1}L$, is the set of all $w$ such that $uw \in L$.
Using it, we define the prefix congruence $\sim_L$ of $L$, as $u \sim_L v$ if $u^{-1}L = v^{-1}L$.
\AP 
We call $L$ ""prefix-independent"" if $\sim_L$ has only one class.
\AP
The ""syntactic FDFA"" $\F_L = (\T_{\sim_L}, \D^L)$ of a "regular $\omega$-language" $L$ is always "saturated", and it is the minimal one (with respect to the size of its "progress DFAs") among all FDFAs that use $\T_{\sim_L}$ as "leading TS".
In general, one may obtain a saturated FDFA with strictly smaller progress DFAs by using a larger leading TS that refines $\sim_L$~\cite[Proposition~2]{Klarlund94}.
This does not hold in "fully saturated" FDFAs, because in those the minimal progress DFAs for all states reached on $\sim_L$-equivalent words coincide.
Intuitively, saturation can be viewed as a form of semantic determinism, which guarantees the existence of unique, minimal "progress DFAs" (for a given "leading TS").
Due to their nondeterministic nature and the resulting non-uniqueness of their progress DFAs, the notion of an "almost saturated" syntactic FDFA does not exist.
%
\section{Saturation Problems}\label{sec:saturation}

In this section, we consider various saturation problems for "FDFA" and "FDWA".
We first establish \PTIME-decidability of saturation for both types of "families".
Then, we illustrate that this can be applied for both active and passive learning of $\omega$-regular languages.
Finally, we consider a relaxed notion of saturation and we show that deciding it is \PSPACE-complete. 

\subsection{Deciding saturation in \PTIME}\label{section:ptimesaturation}
In the following, we state the main theorem regarding both "saturation" and "full saturation".
Subsequently, however, we opt to focus only on the former case of saturation, deferring a proof of the more general statement to~\ifappendixelse{\cref{section:appendixsaturation}}{the full version}.

\begin{theorem}
    For a given "FDFA" $\F$, one can decide in polynomial time whether $\F$ is (fully) "saturated".
    If not, there are (not necessarily) "normalized" pairs $(u,x),(v,y)$ of polynomial length with $ux^\omega = vy^\omega$, such that $\F$ accepts $(u,x)$ and rejects $(v,y)$.%
  \label{theorem:ptimesaturationFDFA}
\end{theorem}

\AP
We call $\F$ ""loopshift-stable"" if shifting the start of the looping part closer to the beginning or further from the beginning preserves acceptance.
Formally, this is the case if, for all "normalized" $(u,ax)$ it holds that $\F$ accepts $(u,ax)$ if, and only if, $\F$ accepts $(ua,xa)$. If the leading TS is trivial, then there is only one progress automaton. In this case we also say that the language of this progress automaton is loopshift-stable (it contains $ax$ if, and only if, it contains $xa$ for each word $x$ and letter $a$).
\AP
We say that $\F$ is ""power-stable"" when (de-)duplications of the loop preserves acceptance, so if for all "normalized" $(u,x)$ holds that either $\F$ accepts all $(u,x^i)$ or $\F$ rejects all $(u,x^i)$, where $i > 0$.
The following result has been established for full saturation in \cite{CianciaV12,CianciaV19} (see the subsection on related work in the introduction).

\begin{lemma}
  An "FDFA" is "saturated" if, and only if, it is "loopshift-stable" and "power-stable".%
  \label{lemma:fdfasaturationcharacterisation}
\end{lemma}
\begin{proof}[Proof (sketch)]
Every "normalized" pair $(u,x)$ has a canonical form that can be obtained by first shifting as much of $u$ as possible into the looping part, and then deduplicating the obtained looping part.
So we can go back and forth between any two representations of the same $\omega$-word by a sequence of only loopshifts and (de)duplications of the looping part.
Thus, $\F$ treats all representations of an $\omega$-word the same if and only if acceptance is preserved under loopshifts and (de)duplications, and the claim is established.
\end{proof}
We now show that deciding the conjunction of "loopshift-stable" and "power-stable" is possible in polynomial time if done in the right order (from the proof of \Cref{theorem:powerclosednessPSPACEcomplete} one can deduce that deciding only "power-stable" alone is \PSPACE-hard). 
\begin{lemma}
    There exists an algorithm which given an "FDFA" $\F$, decides in polynomial time whether $\F$ is "loopshift-stable".
    Moreover, if $\F$ is not "loopshift-stable", the algorithm returns a normalized representation $(u,ax)$ of polynomial length such that $\F$ accepts $(u,ax)$ if and only $\F$ it rejects $(ua,xa)$.
  \label{lemma:rotationinvariancedecidable}
\end{lemma}
\begin{proof}[Proof (sketch)]
It suffices to check for every state $u$ of $\T$ and symbol $a \in \Sigma$, whether there exists a word $w$ such that
$\big(aw$ is accepted by $\D_u$ and loops on $u$ in $\T\big)$
if, and only if,
$\big(wa$ is \emph{not} accepted by $\D_{ua}$ and loops on $ua$ in $\T\big)$.
Finding such a word is equivalent to testing emptiness for a DFA that can be obtained by a product construction between transition systems that are of the same size, and can easily be obtained from $\T$ and $\D_u$.
This guarantees polynomial runtime, and also means that a counterexample of polynomial length exists.
\end{proof}


\begin{lemma}
There exists an algorithm that given a "loopshift-stable" "FDFA" $\F$ in which each "progress DFA" is minimal, decides in polynomial time whether $\F$ is "power-stable".
  Moreover, if $\F$ is not "power-stable", the algorithm returns a counterexample $(u,x)$ of polynomial length such that $\F$ accepts $(u,x^i)$ and rejects $(u,x^j)$ where the exponents $i,j$ are bounded by the maximum size of a progress automaton of $\F$.%
  \label{lemma:rotationinvarianceproperty}
\end{lemma}
\begin{proof}[Proof (sketch)]
Let $u$ be a state in $\T$.
This proof makes use of the assumption that $x$ and $y$ leading to the same state in $\D_u$ implies that $ux$ and $uy$ lead to the same state in $\T$.
Specifically, this property ensures that if a state $q$ of $\D_u$ is reached by a word that loops on $u$ in $\T$, then all words reaching $q$ loop on $u$.

The loopshift-stability of $\F$ guarantees that for all words $x_1, \cdots, x_n, x'_1, \cdots, x'_n$ that loop on state $u$ in $\T$ and satisfy that $\D_u(x_i) = \D_u(x'_i)$ for $1 \leq i \leq n$, we have that $\D_u(x_1\ldots x_n) = \D_u(x'_1\ldots x'_n)$.
This means that in an arbitrary word of the form $x_1 \dots x_n$, swapping out any $x_i$ for some other $x_i'$ that leads to the same state in $\D_u$, will not change the state reached by the resulting word.
Thus, an algorithm which tests for power-stability, only has to consider at most one representative for each state $q$ of a progress DFA $\D_u$.

For each progress DFA $\D_u$, the algorithm iterates though all states $q$ of $\D_u$.
If no words reaching $q$ loop on $u$, it proceeds to the next state as such states are irrelevant when considering normalized acceptance.
Otherwise, all words that reach $q$ loop on $u$, and the algorithm picks a short representative $r$.
The sequence of states reached by the words $r, r^2, \dots$ will repeat after at most $|\D_u|$ entries.
This means that the algorithm only has to check if all of them are accepted, or all of them are rejected.
If the algorithm finds a state in some progress DFA for which this is not the case, then $\F$ is not power-stable.
As the length of $u$ and $r$ are at most $|\T|$ and $|\D_u|$, respectively, it is easy to build a counterexample with the desired properties.
\end{proof}

By first checking for "loopshift-stability", and only then testing for "power-stability" means we can decide "saturation" in polynomial time.
A similar result can be obtained for "FDWA".

\begin{theorem}
  For given "FDWA" $\W$, one can decide in polynomial time if $\W$ is "saturated".
  \label{theorem:ptimesaturationFDWA}
\end{theorem}
\begin{proof}[Proof (sketch)]
As acceptance in an FDWA is naturally invariant under (de)duplication of the looping part, we only have to check whether $\W$ is invariant under loopshifts.
This can be done through an approach similar to the one used in~\cref{lemma:rotationinvariancedecidable} adapted to the $\omega$-semantics of the progress automata of an "FDWA".
\end{proof}%
\subsection{Application to Learning}\label{sec:learning}
In this section we show that the polynomial-time algorithms for (full) saturation of FDFAs can be used to build polynomial time learning algorithms that can learn the "regular $\omega$-languages" represented by FDFAs.
More specifically, we provide a polynomial-time active learner for fully saturated FDFAs, as well as a polynomial-time passive learner that can learn syntactic FDFAs from polynomial data.
As mentioned in the introduction, these are the first polynomial algorithms for representations of the full class of "regular $\omega$-languages".

In general, in automata learning, there is a class $\L$ of potential target languages for which the algorithm should be able to infer a representation in a class $\C$ of automata that represent languages in $\L$. 
\AP
For active learning, we consider the standard setting of a minimally adequate teacher from \cite{Angluin87}, in which an ""active learner"" can ask membership and equivalence queries. 
Membership queries provide information on the membership of words in the target language. An equivalence query can be asked for automata in $\C$, and if the automaton does not accept the target language, then the answer is a counter-example in the symmetric difference of the target language and the language defined by the automaton.
The running time of an "active learner" for the automaton class $\C$ with oracles for a language $L$ as input is measured in the size of the minimal automaton for $L$ (in the class $\C$ under consideration), and the length of the longest counter-example returned by the equivalence oracle during the execution of the algorithm. 
It is by now a well-known result that there are polynomial time active learners for the class of regular languages represented by DFAs \cite{Angluin87,RivestS93,KearnsV94}. Based on such a DFA learner and the full saturation check (\cref{theorem:ptimesaturationFDFA}) we can build an active learner for fully saturated FDFAs.

\begin{restatable}{rsttheorem}{activelearner}
There is a polynomial time "active learner" for the class of "regular $\omega$-languages" represented by "fully saturated FDFAs".
\label{theorem:activelearner}
\end{restatable}
\begin{proof}[Proof (sketch)]
  The learning algorithm uses a polynomial time active learner for DFAs in order to learn a DFA for the language $L_\$ := \{u\$v \mid uv^\omega \in L\}$ over the alphabet $\Sigma \cup \{\$\}$. Membership queries can easily be answered using the membership oracle for ultimately periodic words. On equivalence queries of the DFA learner, our algorithm first translates the DFA over $\Sigma \cup \{\$\}$ into an FDFA $\F_\A$ with $\ReprFDFA{\F_\A} = \{(u,v) \in \Sigma^* \times \Sigma^+ \mid u\$v \in \LangDFA\A\}$, which is a straight-forward construction. 
  Then it checks if the FDFA is "fully saturated" (\cref{theorem:ptimesaturationFDFA}), and if not returns a counter-example to the DFA learner derived from the example for failure of full saturation.
  If the FDFA is fully saturated but there is a counter-example $(u,v)$ this is passed on as $u\$v$ to the DFA learner. 
\end{proof}

A ""passive learner"" gets as input two sets $S_+$ of examples that are in the language and $S_-$ of examples that are not. 
It outputs an automaton from $\C$ such that all examples from $S_+$ are accepted and all examples from $S_-$ are rejected. The pair $S = (S_+,S_-)$ is referred to as sample. It is an $L$-sample if all examples in $S_+$ are in $L$, and all examples in $S_-$ are outside $L$. For the class of "regular $\omega$-languages" represented by their syntactic FDFAs, the sets $S_+$ and $S_-$ contain representations of ultimately periodic words, and the passive learner outputs a syntactic FDFA $\F$ that is consistent with $S$, that is, $S_+ \subseteq \ReprFDFA\F$ and $S_- \cap \ReprFDFA\F = \emptyset$.
In the terminology of \cite{Gold78}, we say that a passive learner $f$ can learn automata from $\C$ for each language in $\L$ in the limit if for each $L \in \L$ there is an $L$-sample $S_L$, such that $f(S_L)$ returns an automaton $\A$ that accepts $L$, and for each $L$-sample $S$ that extends $S_L$ (contains all examples from $S_L$), $f(S)$ returns the same automaton $\A$. Such a sample $S_L$ is called a characteristic sample for $L$ and $f$. Further, $f$ is said to learn automata from $\C$ for each language in $\L$ in the limit from polynomial data if for each $L$ there is a characteristic sample for $L$ and $f$ of size polynomial in a minimal automaton for $L$ (from the class $\C$).

\begin{restatable}{rsttheorem}{passivelearner}
There is a polynomial time "passive learner" for "regular $\omega$-languages" represented by syntactic FDFAs that can infer a "syntactic FDFA" for each "regular $\omega$-language" in the limit from polynomial data.
\label{theorem:passivelearner}
\end{restatable}
\begin{proof}[Proof (sketch)]
  The passive learner constructs an FDFA $\F$ from the given sample $S$ using techniques known from passive learning of DFAs. It then checks whether $\F$ is a "syntactic FDFA" that is consistent with $S$ and returns $\F$ if yes.
  Otherwise, it returns a default FDFA that accepts precisely all representations of the ultimately periodic words in $S_+$. The test whether $\F$ is a "syntactic FDFA" can then be done by checking whether $\F$ is saturated (\cref{theorem:ptimesaturationFDFA}). We guarantee in step 1 that there is no smaller leading congruence for the given sample. 
\end{proof}
\subsection{Almost Saturation}\label{section:almostsaturation}
"Saturation" is sufficient, but not necessary to guarantee $\omega$-regular semantics.
In \cite{LiST23}, the weaker notion of \emph{almost saturation} is shown to also be sufficient, and in this section we consider the complexity of deciding whether an "FDFA" has this property.
\AP
We say that $\F = (\T, \D)$ is ""almost saturated"" if $(u,x) \in \NormFDWA{\F}$ implies $(u,x^i) \in \NormFDWA{\F}$ for all $i > 1$.
In other words, once such an FDFA accepts a loop, it must also accept all repetitions of it, which means that the languages accepted by the progress automata are closed under taking powers of the words in the language.

\begin{restatable}{rsttheorem}{almostsaturationPSPACEcomplete}
  It is \PSPACE-complete to decide whether an "FDFA" is "almost saturated".%
  \label{theorem:powerclosednessPSPACEcomplete}
\end{restatable}
\begin{proof}[Proof (sketch)]
    Membership in \PSPACE can be established by an algorithm that guesses a word showing that the FDFA is not almost saturated.
    To show \PSPACE-hardness, we reduce from the DFA intersection problem, which asks whether the intersection of an arbitrary sequence $\D_1,\dotsc,\D_p$ of DFAs is non-empty and is known to be \PSPACE-complete~\cite{Kozen77}.
    We assume that the number $p$ of automata is prime, otherwise we pad the sequence with universal DFAs.
    We introduce a fresh symbol $\#$ and construct the "FDFA" $\F = (\T, \A)$ where $\T$ is trivial and $\A$ is a DFA for the complement of $\# \LangDFA{\D_1} \# \LangDFA{\D_2} \# \dotsi \# \LangDFA{\D_p}$. Then a word $x$ is accepted by $\A$ while $x^i$ is rejected iff $x=\#u$ for $u$ accepted by all the $\D_j$.
\end{proof}%
%
\knowledge{notion}
| UP-language
\knowledgenewrobustcmd{\synL}{\ensuremath{\mathrel{\cmdkl{\approx_P}}}}
\knowledgenewrobustcmd{\synN}{\ensuremath{\mathrel{\cmdkl{\approx_\N}}}}
\knowledgenewrobustcmd{\tpN}{\ensuremath{\mathop{\cmdkl{\tau_\N}}}}
\knowledge{notion}
| accepting
\knowledge{notion}
| rejecting
\knowledge{notion}
| terminal@tp
\knowledgenewrobustcmd{\Ter}{\ensuremath{\mathop{\cmdkl{\op{Ter}}}}}
\knowledge{synonym}
| terminal word
| terminal words
| terminal@word
\knowledgenewrobustcmd{\Ret}{\ensuremath{\mathop{\cmdkl{\op{Ret}}}}}
\knowledge{notion}
| good witness
\knowledge{notion}
| FNFA

\section{Regularity}\label{section:regularity}
We only consider normalized representations in this section, following other works on FDFAs~\cite{AngluinBF18,LiCZL21,LiST23}, and provide proof details in~\ifappendixelse{\cref{section:appregularity}}{the full version}. We believe that the methods that we present in this section  also work for the non-normalized semantics, but this requires some extra details, so we prefer to stick to one setting.

An FDFA or FDWA $\F$ defines a \AP ""UP-language"" (a set of ultimately periodic words) through the representations it accepts.
If this UP-language $L$ is "UP-regular", that is, there is an $\omega$-regular language $L'$ such that $L = \UP{\Sigma}\cap L'$, then $L'$ is unique and $\F$ defines this $\omega$-regular language.
For FDFAs it is known that they define $\omega$-regular languages if they are saturated \cite{AngluinBF18} or almost saturated \cite{LiST23}.
But in general, the "UP-language" of an FDFA needs not to be "UP-regular", which is witnessed e.g. by an FDFA accepting pairs of the form  $(u,ba^i)$ for which the "UP-language" is $\bigcup_{i \in \mathbb{N}}\left(\Sigma^*(ba^i)^\omega\right)$, which is not "UP-regular" \cite[Example~2]{LiST23}.

The main result in this section is that the \emph{regularity problem for FDFAs} is decidable: given an FDFA, decide whether it defines a regular $\omega$-language, that is, whether its "UP-language" is "UP-regular". We show that the problem is \PSPACE-complete. Note that "almost saturation" is \emph{not} a necessary condition, as witnessed by a simple progress DFA accepting $a^i$ for all odd integers $i$, which clearly defines a "UP-regular" language for the trivial "leading TS" but is not almost saturated.

But let us start with the observation that for "FDWAs" there is no decision problem because the "UP-language" is always "UP-regular", which can be shown by the same construction to NBAs that is used for (almost) saturated FDFAs (which goes back to \cite{CalbrixNP93}).

\begin{restatable}{rstlemma}{rstFDWAregular} \label{lemma:FDWAregular}
  For every "FDWA" $\W = (\T,\B)$, the "UP-language" $\{ux^\omega \mid x^\omega \text{ is accepted by } \B_u\}$ is "UP-regular"
\end{restatable}


%

We now turn to the regularity problem for FDFAs.  We first show that it is sufficient to analyze "loopshift-stable" FDFAs with trivial leading TS (\Cref{lem:stabalize}, Theorem~\ref{thm:regular}). However, this requires to work with \emph{nondeterministic} progress automata because making an FDFA loopshift-stable can cause an exponential blow-up.

\AP
We call $\F=(\T,\N)$ a \emph{family of NFAs} (""FNFA""), where $\T$ is the leading \emph{deterministic} TS and, for each $q \in \states{\T}$, $\N_q$ is an "NFA".
Recall, that we consider the normalized setting, so $(u,v)$ is accepted by $\F$ if $\T(u) = \T(uv)$ and $v$ is accepted by $\N_{\T(u)}$.
FDFAs are simply special cases of FNFAs, where all progress automata are deterministic.
The notion of "loopshift-stability" and other notions defined for FDFA naturally carry over to FNFAs as well.

\begin{restatable}{rstlemma}{rststabilize}
  \label{lem:stabalize}
  Given an 
  "FDFA" or "FNFA" $\F = (\T,\N)$, we can build a "loopshift-stable"
  FNFA $\F'=(\T,\N')$ in polynomial time such that $\F$ and $\F'$ define the same "UP-language".
\end{restatable}

\begin{proof}[Proof (sketch)]
  $\F'$ can nondeterministically guess, for $(u,v)$, the shift $(uv_1,v_2v_1)$; $\T(uv_1)$; and the state $p$ the accepting run of $N_{\T(uv_1)}$ on $v_2v_1$ is in after reading $v_2$, and use its nondeterministic power to validate the guess.
\end{proof}

We now state the main theorems of this section.

\begin{theorem}
  \label{thm:regular}
  The regularity problem for FDFAs is \PSPACE-complete. More precisely,
  the following problems are \PSPACE-complete:
  \begin{enumerate}
    \item Is the "UP-language" of a given "FNFA" $\F=(\T,\N)$ "UP-regular"?
    \item Is the "UP-language" of a given "FDFA" $\F=(\T,\D)$ with a trivial TS $\T$ "UP-regular"?
    \item Is the "UP-language" of a given "FNFA" $\F=(\T,\N)$ with a trivial TS $\T$ and a "loopshift-stable" set $P=\LangDFA{\N}$ "UP-regular"?
    \item Is the "UP-language" of a given "FNFA" $\F=(\T,\N)$ with a trivial TS $\T$ "UP-regular"?
  \end{enumerate}
\end{theorem}

\begin{proof}[Proof outline:]
  To show hardness, it is enough to show hardness for (2), as hardness for (3) then follows from \cref{lem:stabalize}, and both (2) and (3) are special cases of (1) and (4). We show hardness of (2) in \Cref{thm:regularity-hard}.

  Most of this section is dedicated to establish a decision procedure (which is in \PSPACE) for (3).
  This is done in Lemmas \ref{lem:regularEquivalent} through \ref{lem:InPSPACE}.
  Using \cref{lem:stabalize} reduces (4) to (3), and (2) is a special case of (4). Case (1) can be reduced to case (3) by considering a labeling of the words with the states of the leading TS.
\end{proof}

\noindent\textbf{Prefix-independent languages.}
For the \PSPACE upper bound, as justified by the reduction to (3) of \Cref{thm:regular}, we focus on a "loopshift-stable" FNFA $\F = (\T,\N)$ with trivial $\T$. So there is only one progress NFA $\N = (\Sigma, Q, \delta, \iota, F)$, which accepts some language $P = \LangDFA{\N}$ satisfying $uv \in P \Leftrightarrow vu \in P$ as $\F$ is "loopshift-stable".
Since $\T$ is trivial, the "UP-language" $L$ of $\F$ is "prefix-independent" and is of the form $L = \{u\cdot v^\omega \mid (u,v) \in \NormFDFA{\F}\} = \Sigma^* \cdot \Pow{P}$ (recall that $\Pow{P} = \{v^{\omega} : v \in P\}$). Our goal is to decide whether $L$ is "UP-regular".

As the semantics is existential, it can happen that some representations of  a word $v^\omega \in \Pow{P}$ are not in $P$.
This makes the problem non-trivial even in the "prefix-independent" case as witnessed by the language $P = \{a^iba^j \mid i+j \ge 1\}$, which consists of all words with exactly one $b$ and at least one $a$ over the alphabet $\Sigma = \{a,b\}$.
Clearly, $P$ is regular, but $L = \Sigma^* \cdot \{\left(a^iba^j\right)^\omega \mid i+j \ge 1\}$ is not "UP-regular".

\AP
We now define a congruence relation $\synL$: for $x, y \in \Sigma^*$,
\[ x \intro*\synL y \text{ iff for all } v\in\Sigma^* \text{ holds } (xv)^{\omega} \in \Pow{P} \Longleftrightarrow (yv)^{\omega} \in \Pow{P}. \]
As $L = \Sigma^* \cdot \Pow{P}$ is "prefix-independent", $\synL$ is the syntactic congruence of $L$ defined in~\cite{Arnold85}, and also corresponds to the definition of the progress congruence of the syntactic FORC \cite{MalerS97} and to the periodic progress congruence from \cite{AngluinF16} (where the latter two are only defined for regular $\omega$-languages). It follows from results in \cite{CalbrixNP93,AngluinBF18} that $L$ is "UP-regular" if, and only if, $\synL$ has finite index.

\begin{restatable}{rstlemma}{rstregularEquivalent}
  \label{lem:regularEquivalent}
  Let $\F=(\T,\N)$ be a "loopshift-stable" FNFA with trivial TS $\T$; let $P=\LangDFA{\N}$ and let $L = \Sigma^* \cdot \Pow{P}$.
  Then the following claims are equivalent:
  (1) $\synL$ has finite index; and (2) the "UP-language" $L$ of $\F$ is "UP-regular".
\end{restatable}

\noindent\textbf{Transition profiles.}
\AP
By \Cref{lem:regularEquivalent}, to determine whether or not the UP-language of $\F$ is UP-regular, we only need to check whether $\synL$ has finite index.
To this end, we introduce the ""transition profile"" of a finite word $x \in \Sigma^+$, which captures the behavior of $x$ in a TS or automaton, in this case the progress NFA $\N$.
Formally, $\tau$ is given as a mapping $\tau : Q \to Q$ if $\N$ is deterministic and $\tau : Q \to 2^Q$ if it is nondeterministic, with $q \mapsto \delta^*(q, x)$, and we write $\intro*\tpN$ for the function that assigns to a finite word its transition profile in $\N$.
There are at most $|Q|^{|Q|}$ and $2^{|Q|^2}$ different mappings for deterministic and nondeterministic automata, respectively.
\AP
We now refine $\synL$ using the "transition profile" by introducing
\[x \intro*\synN y \text{ if, and only if, }  x \synL y \text{ and } \tpN(x) = \tpN(y)\ .\]
Clearly, $\synL$ has finite index if, and only if, $\synN$ has finite index.

"Transition profiles" are useful, because they offer sufficient criteria for a word $x$ to be in the power language: we can directly check from $\tpN(x)$ whether or not there is some power $j>0$ such that $x^j \in P$.
This, of course, entails $x^\omega \in \Pow{P}$.
We call such transition profiles \AP ""accepting"" and transition profiles, where no power of $x$ is accepted, ""rejecting"". Note that $\tpN(x)$ is called "accepting" if \textit{some power} of $x$ is accepted, not necessarily $x$ itself. Correspondingly,  $\tpN(x)$ is "rejecting" if \emph{all powers} of $x$ are rejected.

Note that there can be words $x$ such that $x^\omega \in \Pow{P}$ but no power of $x$ is in $P$.
For instance, let 
$P = \{b^iab^j : i+j \geq 0\}$
be a loopshift-stable regular language:
clearly $(ab\cdot ab)^{\omega} \in \Pow{P}$ but $(ab \cdot ab)^i \notin P$ for all $i > 0$.

Recall that the root $\wordroot{x}$ is the shortest word $u$ such that $u^{\omega}  = x^\omega$.
To get a feeling for the properties of transition profiles, for a word $x$ with $x=\wordroot{x}$, $x^\omega \in \Pow{P}$ can only hold if some power of $x$ is in $P$, which we can read from $\tpN(x)$.

\AP
We call a transition profile $\tau$ ""terminal@@tp"" if it is "accepting", but $\tau^i$ is "rejecting" for some power $i$.
(Note that $\tpN(x^i)$ is the function obtained from iterating $\tpN(x)$ $i$ times. However, we can take the power of a transition profile directly without knowing $x$, and even for transition profiles that do not correspond to a word.)

\AP
We call a word $x$ such that $\tpN(x)$ is "terminal@@tp", a ""terminal word"" and write $\Ter(P)$ for the set of terminal words in $P$. So $x$ is a "terminal word" if some power $x^j$ is in $P$, and there is $i > 1$ such that $x^{ik} \notin P$ for all $k \ge 1$. Note that if $x$ is a terminal word, then so is $\wordroot{x}$.
The following lemma expresses finite index of $\synL$ in terms of a property of $\Ter(P)$ that we use in our decision procedure.

\begin{restatable}{rstlemma}{rstfinfin}
  \label{lem:fin2fin}
  If the NFA $\N$ accepts a "loopshift-stable" language $P$, then $\synN$ (and thus $\synL$) have finite index if, and only if, $\Ter(P)$ has finitely many roots.
\end{restatable}

\begin{proof}[Proof (sketch).]
  We show that, if there are infinitely many roots in $\Ter(P)$, then we can construct from them an arbitrary number of words that are distinguished by $\synL$; $\synL$ therefore has infinite index.

  Conversely, if there are finitely many roots $r_1,\ldots,r_n$ in $\Ter(P)$, then two words $x,y$ are identified equivalent by $\synL$ if they define the same transition profile $\tpN(x) = \tpN(y)$ and for all $z \in \Sigma^*$ $\tpN(xz)$ is accepting
  or there is a root $r_i \in \Ter(P)$ such that $(xz)^{\omega} = r_i^{\omega}$ iff there is a root
  $r_j \in \Ter(P)$ such that $(yz)^{\omega} = r_j^{\omega}$ where $1 \leq i, j \leq n$.

  These are all regular properties, and $\synL$ is therefore finite.
\end{proof}

We now show that infinitely many roots of $\Ter(P)$ are witnessed by transition profiles with specific properties, called good witnesses, as defined below.

We say that a word $x$ visits $\tau_g$ first if (a) $\tpN(x)= \tau_g$ holds and (b) no true prefix $y$ of $x$ satisfies $\tpN(y) = \tau_g$.
If there is a word $x$ that visits $\tau_g$ first, we say that $u$ recurs to $\tau_g$ if $u \neq \varepsilon$ and (a) $\tpN(xu)= \tau_g$ holds and (b) no true non-empty prefix $v$ of $u$ satisfies $\tpN(xv) = \tau_g$.

\AP
We say that $\tau_g$ is a ""good witness"" if it is "terminal@@tp" and one of the following holds:
\begin{enumerate}
  \item there are infinitely many words $x$ that visit $\tau_g$ first, \emph{or}
  \item there is a word $x$ that visits $\tau_g$ first and a word $u$ that recurs on $\tau_g$ that have different roots.
\end{enumerate}

\begin{restatable}{rstlemma}{rstInfiniteGoodness}
  \label{lem:InfiniteGoodness}
  Let $P = \LangDFA{\N}$. There is a "good witness" $\tau_g$ if, and only if, $\Ter(P)$ has infinitely many roots.
\end{restatable}

\begin{proof}[Proof (sketch).]
  If $\Ter(P)$ has infinitely many roots, then infinitely many of them refer to one witness. We show that they cannot all be constructed from one word $x$ that visits $\tau_g$ first and one word $u$ that recurs on $\tau_g$ \emph{and} has the same root as $\tau_g$, so that one of the two cases must hold.

  Reversely, if $\tau_g$ is a "good witness" we distinguish case (1), where we argue that the infinitely many words that visit $\tau_g$ first must have different roots (and roots of "terminal words" are "terminal@@tp").
  For the other case, we build infinitely many roots with the terminal transition profile $\tau_g$ from a word $x$ that visits $\tau_g$ first and a word $u$ that recurs on $\tau_g$ and has a different root than $x$.
\end{proof}

Checking if there is a "good witness" can be done on the succinctly defined TS whose state space are the transition profiles that tracks the transition profiles of finite words.

\begin{restatable}{rstlemma}
  {rstIsItGood}
  \label{lem:InPSPACE}
  For a given NFA $\N$, we can check if there is a "good witness" in \PSPACE.
\end{restatable}

\begin{proof}[Proof (sketch)]
  We observe that tracing the transition profiles is done by a TS $\T'$ of size exponential in $\N$, but succinctly represented by $\N$ (and in $\tau_g$ for checking that $\tau_g$ is "terminal@@tp").
  Checking the conditions of a good transition profile are a number of reachability problems in this transition profile, and all of them are in NL the size of $\T'$.
\end{proof}

This finishes the proof for the \PSPACE upper bound. We now show the lower bound for case (2) of \Cref{thm:regular}.

\begin{theorem} \label{thm:regularity-hard}
  Deciding for a DFA $\D$ whether $\Ter(P)$ is finite for $P=\LangDFA{\D}$ and deciding whether $\Ter(P)$ has finitely many roots are \PSPACE-hard. Furthermore, deciding if the UP-language of an FDFA with trivial leading TS is UP-regular, is \PSPACE-hard.
\end{theorem}

\begin{proof}
  Assume that we are given $p$ DFAs $\D_1, \D_2, \cdots, \D_p$ that do not accept the empty word; we assume that $p$ is prime (and otherwise pad) and that $\#\notin \Sigma$.
  We define $L_1=\#^+\Sigma^+\#^*$, $L_p={L_1}^p$, and $L_\cap = \bigcap_{i=1}^p \LangDFA{\D_i}$.
  Deciding if $L_\cap = \emptyset$ is the \PSPACE-hard problem we reduce from.

  We build two DFAs $\D$ and $\D'$ using the states of $\D_1,\ldots,\D_p$ plus $p+1$ extra states.
  $\D'$ accepts a word $w$ if it is in $L_p$ and of the form $w=\#^+w_1\#^+w_2\ldots\#^+w_p\#^*$ s.t.\ $w_i\in \Sigma^+$ holds for all $i \leq p$ \emph{and} $w_i \notin \LangDFA{\D_i}$ holds for some $i\leq p$. We call its language $L'$.

  $\D$ accepts a word $w$ if $w\in {L_p}^*\cdot L'\cdot {L_p}^*$ \emph{or} $w \in {L_1}^i$ and $p$ does not divide $i$.

  If $L_\cap \not= \emptyset$, then every word $v\in \#^+L_\cap$ is "terminal@@word" for $P=\LangDFA{\D}$, because $v^p$ defines a "rejecting" transition profile.
  These words are all roots, and the $\#^+$ part ensures that there are infinitely many of them.

  For the case $L_\cap = \emptyset$ we assume towards a contradiction that $\Ter(P)$ contains a root $u$. Then $u \in {L_1}^i\#^*$ for some $i$.
  If $p$ divides $i$, then either all powers of $u$ are accepted or no power of $u$ is accepted (contradiction). If $p$ does not divide $i$, then $u^j$ is accepted if $p$ divides $j$ (because all $i$ words in $\Sigma^+$ are rejected by some $\D_k$ due to $L_\cap = \emptyset$), and $u^j$ is accepted if $p$ does not divide $j$, because then $p$ does not divide $ij$ and $u^j \in {L_1}^{ij}$ (contradiction).

  For showing the second part of the theorem, we use the DFA $\D$ as progress automaton for an FDFA with trivial leading TS, and then use \Cref{lem:regularEquivalent} and \Cref{lem:fin2fin}. These lemmas require the language that is accepted to be "loopshift-stable". While $\D$ does not recognize a "loopshift-stable" regular language $P$, it is not hard to see that the "loopshift-stable" language $P' = \{xy \mid x,y \in (\Sigma\cup\{\#\})^*, yx \in P\}$ we obtain by the construction from \cref{lem:stabalize} also satisfies that $\Ter(P')=\emptyset$ if $L_\cap = \emptyset$. Since $\Ter(P) \subseteq \Ter(P')$, we also get that $\Ter(P')$ contains infinitely many roots if $L_\cap \not= \emptyset$.
  Furthermore $P'$ and $P$ define the same "UP-language", and thus deciding if this "UP-language" is "UP-regular" is the same as deciding if $\Ter(P)$ has finitely many roots.
\end{proof}

%
%
\section{Comparison}
\label{section:comparison}
In this section, we analyze the influence that the properties introduced in~\cref{sec:saturation} can have on the succinctness of the various models.
Moreover, we also consider the transformation between models with certain properties.
The results of this section are presented visually in the overview in~\cref{fig:overview}.
We sketch the families used to obtain lower bounds only roughly, deferring formal definitions and proofs~\ifappendixelse{to the appendix}{the full version} for the sake of readability.

"Saturated FDWAs" may be exponentially more succinct than "saturated FDFAs".
Intuitively, this comes from the fact that by~\cref{lemma:fdfasaturationcharacterisation} FDFAs have to be "loopshift-stable", which implies that their "progress DFAs" exhibit a monoidal structure.
We have described this structure in more detail, and shown how it makes deciding "power-stability" tractable in~\cref{lemma:rotationinvarianceproperty}.
However it also means that if the "leading TS" is trivial, then in the "progress automaton" two words $u,v$ may only lead to the same state if they cannot be distinguished in arbitrary products, that is, for all words $x,w$, either $xuw$ and $xvw$ are both accepted by the progress DFA, or both are rejected.
This insight allows us to obtain the following lower bound with standard techniques.

\begin{restatable}{rstlemma}{lowerboundCombinedToSatFDFA}
    For each $n > 0$, there is a language $L_n$ that can be recognized by a "saturated FDWA" of "size@@family" $(1,n+2)$, and by an "almost saturated FDFA" of "size@@family" $(1,n+2)$, but for every "saturated" "FDFA" of "size@@family" $(1,m)$ that recognizes $L_n$, we have that $m \geq n^n$.
  \label{lemma:lowerbound:combinedToSatFDFA}
\end{restatable}
\begin{proof}[Proof (sketch)]
    We represent mappings $\{1,\dotsc,n\}\rightarrow\{1,\dotsc,n\}$ as words over an alphabet of constant size.
    The language $L_n$ contains a word $w$ if it has infinitely many infixes $\#u\#$ such that $1$ is a fixpoint of the mapping represented by $u$.
    A "saturated FDWA" can wait for an occurrence of $\#$ and then keep track of the number that $1$ is mapped to, which is possible with $n+2$ states.
    As the progress automaton of this FDWA has only a single accepting state which is a sink, it is "almost saturated" when viewed as "FDFA".
    By~\cref{lemma:fdfasaturationcharacterisation}, a "saturated FDFA" must be "loopshift-stable", and therefore cannot wait for an occurrence of $\#$ before following a single number.
    Thus, it has to memorize the full mapping encoded by the word read so far, for which it needs at least $n^n$ states.
\end{proof}

The size of an "almost saturated FDFA" is, in general, incomparable to the size of a "saturated FDWA", and there may be an exponential blowup in either direction, which the following lemma formalizes.

\begin{restatable}{rstlemma}{incomparableSatFDWAalsatFDFA}
  \label{lemma:incomparableSatFDWAalsatFDFA} 
  For all $n > 0$, there are languages $L_n$ and $L_n'$
  such that
  \begin{enumerate}
      \item\label{satFDWAsmall:alsatFDFAbig}
      $L_n$ is accepted by a "saturated FDWA" of "size@@family" $(1,2n+1)$, while for every "almost saturated FDFA" for $L_n$ with trivial "leading TS" has a progress automaton of size $2^{\frac{n}{2}}$
      \item\label{alsatFDFAsmall:FDWAsmall:satFDWAbig}
      $L_n'$ can be recognized by an "almost saturated FDFA" of size $(1,n+3)$, and by an "FDWA" of size $(1,n+3)$, but every "saturated FDWA" with trivial "leading TS" has at least $2^n$ states
      \item\label{alsatFDFAcomplementbig}
      in every "almost saturated FDFA" with trivial "leading TS" that accepts the complement of $L_n'$, the progress automaton has at least $\frac{2^n}{2n}$ states.
  \end{enumerate}
\end{restatable}
\begin{proof}[Proof (sketch)]
    The alphabet of $L_n$ consists of non-empty subsets of $\{1,\dots,2n\}$.
    A word is in $L_n$ if some number $i$ is member of only finitely many symbols, meaning $L_n$ is "prefix-independent".
    Due to its limit behavior, a "saturated FDWA" can test on the loop for an occurrence of each number in sequence, leading to a progress automaton with $2n+1$ states.
    But one can show that for all $X,Y \subseteq \{1,\dots,2n\}$ with $X \not= Y$ and $|X| = |Y| = n$, either $X$ and $Y$, or their complements must lead to different states from the initial state of the progress automaton of any "almost saturated FDFA" for $L_n$.

    The language $L_n'$ is defined over the alphabet $\{0,1\}$, and contains all words with infinitely many infixes $0u0$ such that $|u| = n-1$, which is clearly "prefix-independent".
    When using a trivial "leading TS", an "FDWA" or "almost saturated FDFA" can simply guess the start of an infix $0u0$ and verify that $|u| = n-1$, which can be done with $n+3$ states.
    This is not possible in a "saturated FDWA" since it has to be invariant under loopshifts.
    Similarly, in the progress automaton of any "almost saturated FDFA" for the complement of $L_n'$, distinct words of length $n$ must lead to different states.
\end{proof}

The following bounds were already known in the literature (c.f.~\cite[Section~5]{Klarlund94} and~\cite[Theorem~2]{AngluinF16}).
We include them for the purpose of completeness, and provide families of languages achieving them in a systematic way.

\begin{restatable}{rstlemma}{expgapsatFDFAsynFDFAfullysatFDFA}
    \label{lemma:expgapsatFDFAsynFDFAfullysatFDFA}
    For all $n > 0$ there are languages $L_n, L_n'$ such that
    \begin{enumerate}
        \item $L_n$ is recognized by a "saturated FDFA" of size $(n, 2n)$, but the "syntactic FDFA" for $L_n$ is of size $(1, n^n)$, and
        \item the "syntactic FDFA" for $L_n'$ is of size $(n+1, 2n + 1)$ and every "fully saturated FDFA" for $L_n'$ has a progress DFA of size at least $n^n$.
    \end{enumerate}
\end{restatable}
\begin{proof}[Proof (sketch)]
We consider words of the form $u_1\#u_2\#\dots$ where each $u_i$ encodes a mapping $m_{u_i}$ from $\{1,\dots,n\}$ to itself, akin to~\cref{lemma:lowerbound:combinedToSatFDFA}.
Such a word is in $L_n$ if infinitely many $u_i$ represent a mapping such that $1$ is a fixpoint.
A small "saturated FDFA" $(\T, \D)$ tracks in its leading TS $\T$ the number that $1$ is sent to under the mapping encoded by the word read since the last $\#$.
This leads to $n$ progress DFAs of size $2n$ each, and $\D_i$ accepts words of the form $x\#y$ where $m_x$ maps $i$ to $1$ and $m_y$ maps $1$ to $i$.
The "syntactic FDFA" on the other hand has a trivial leading TS and in its progress DFA, any words encoding distinct mappings have to lead to different states.
For the second claim, we slightly modify $L_n$ to force the leading TS of the "syntactic FDFA" to be roughly of the shape of $\T$, which means the each progress DFA is of linear size.
A fully saturated FDFA for $L'$, however, cannot simply ignore words if they do not loop on a state of the leading TS, which forces it to track all $n^n$ mappings.
\end{proof}
%
%
\section{Conclusion}
\label{section:conclusion}
We have determined the complexity of several decision problems related to the saturation of FDFAs and FDWAs.
We showed that this has practical implications using polynomial-time (full) saturation check to provide a polynomial-time active learner for fully saturated FDFAs, as well as a polynomial-time passive learner capable of learning in the limit the syntactic FDFA for each regular $\omega$-language from polynomial data.
These are the first algorithms of their kind for representations of the entire class of $\omega$-regular languages.
Beyond these direct applications, we believe that our proofs offer deeper insights into the structure of FDFAs, particularly regarding the reasons for non-saturation and non-regularity, which may prove useful for further results and algorithms for $\omega$-languages represented by families of automata.%
\bibliography{references}
\ifappendix{%
  \newpage
  \appendix
  \section{Full proof of \PTIME-decidability of (full) saturation}
\label{section:appendixsaturation}
\newcommand{\X}{\ensuremath{\mathcal{X}}}
In this section we give full proofs for the \PTIME decidability of (full) saturation for FDFAs and saturation for FDWAs.
\AP
In order not to have to state all the lemmas twice, once for saturation and once for full saturation, we give a general formulation of the lemmas with respect to a ""reference set"" $\X \subseteq \Sigma^* \times \Sigma^+$, which either consists of all pairs or only the normalized ones for the "family" under consideration.
Note that the normalized pairs depend on $\F$, so when we say that $\X$ is a reference set, it is always w.r.t.\ some $\F$. This $\F$ should always be clear from the context.
The following properties are used in the sequel and can be shown by a simple case distinction. Recall that we identify each state of a transition system/automaton with the length-lexicographically minimal word leading to that state.
\begin{lemma}
    Let $\X$ be a "reference set" for some family $\F = (\T, \D)$, then
    \begin{enumerate}
        \item $(\T(u), x) \in \X$ if and only if $(u, \D_u(x)) \in \X$ if and only if $(u,x) \in \X$
        \item $(u,x) \in \X$ implies $(u,x^i) \in \X$ for all $i > 0$
        \item $(u, ax) \in \X$ implies $(ua, xa) \in \X$
        \item $(u,x), (u,y) \in \X$ implies $(u,xy) \in \X$
        \item for each $u$, there is a DFA of size $|\T|$ that accepts precisely those $x$ such that $(u,x) \in \X$.
    \end{enumerate}
    \label{lem:refset}
\end{lemma}
\begin{proof}
    Surely, if $\X$ is $\Sigma^* \times \Sigma^+$, the properties trivially hold.
    Otherwise, we have that $(u,x) \in \X$ if and only if $ux \sim_\T u$, or equivalently $\T(ux) = \T(u)$.
    \begin{enumerate}
        \item 
        Obviously, $\D_u(x)$ and $x$ lead to the same state in $\D_u$ and so by our assumption on the progress automata, $u\D_u(x) \sim_\T ux$.
        Because $\D_u = \D_{\T(u)}$ and $\sim_\T$ is a right congruence, we obtain $u\D_{\T(u)}(x) \sim_\T \T(u)\D_{\T(u)}(x) \sim_\T u\D_u(x) \sim_\T ux$, from which the equivalence immediately follows.
        \item
        For all $i > 0$, we have $ux^i = uxx^{i-1} \sim_\T ux^{i-1} \sim_\T \dots \sim_\T ux \sim_\T u$.
        \item 
        If $u\cdot ax \sim_\T u$, then since $\sim_\T$ is a right congruence, surely also $ua \cdot xa = u \cdot ax \cdot a \sim_\T u\cdot a$.
        \item 
        If $ux\sim_\T u$ and $uy\sim_\T u$, then because $\sim_\T$ is a right congruence surely also $uxy\sim_\T uy \sim_\T u$.
        \item
        If $\X$ is universal, such a DFA is universal.
        Otherwise, we obtain it for some $u$ from $\T$ by setting $\T(u)$ as the initial state and $\{\T(u)\}$ as final states.
    \end{enumerate}
\end{proof}

A "family" $\F = (\T, \D)$ and a reference set $\X$ give rise to an equivalence relation on representations.
We define that $(u,x) \intro*\redequiv^\X_\F (v,y)$ if $\{(u,x),(v,y)\} \subseteq X$ implies either both $(u,x)$ and $(v,y)$ are accepted by $\F$ or both are rejected by $\F$.
Conversely, if $(u,x) \not\redequiv^\X_\F (v,y)$, then we must have that both are in $\X$ and $\F$ accepts one but not the other.
Note in particular, that if $\X$ is not universal, then any non-normalized representation is $\redequiv^\X_\F$-equivalent to all representations.
We can now use this relation to define the notions of loopshift- and power-stability as well as that of saturation in a way that is agnostic to whether we consider only normalized or all representations.

\AP
$\F$ is called ""saturated for $\X$"" if for all representations $(u,x),(v,y) \in \X$ with $ux^\omega = vy^\omega$ holds $(u,x) \redequiv^\X_\F (v,y)$.
From this, we can recover the notions used in the main part.
Specifically, $\F$ is "saturated" if $\F$ is saturated for the set of "normalized" representations and "fully saturated" if $\F$ is saturated for all representations.
\AP
Further, we call $\F$ is called ""loopshift-stable for $\X$"" if for all $(u,ax) \in \X$ holds $(u, ax) \redequiv^\X_\F (ua, xa)$.
Finally, $\F$ is ""power-stable for $\X$"" if for all $(u,x) \in \X$ holds $(u,x^i) \redequiv^\X_\F (u,x^j)$ for all $i, j > 0$.%

\begin{lemma}
  Let $\F = (\T, \D)$ be an "FDFA" and $\X$ be a "reference set".
  $\F$ is "saturated for $\X$" if and only $\F$ it is
  "loopshift-stable for $\X$" and "power-stable for $\X$".%
  \label{app:lemma:fdfacharacterisation}
\end{lemma}
\begin{proof}
  For the direction from left to right, assume that $\F$ is saturated for $\X$ and let $(u,ax), (v,y) \in \X$.
  Then by applying~\cref{lem:refset}, we obtain $(v,y^i) \in \X$ for all $i > 1$ from (2) and $(ua,xa)\in\X$ from (3).
  Since additionally, $u(ax)^\omega = ua(xa)^\omega$ and $vy^\omega = v(y^i)^\omega$ for $i > 0$, we  immediately get that $\F$ it is "loopshift-stable for $\X$" and "power-stable for $\X$" from the fact that $\F$ is "saturated for $\X$".

  For the opposite direction, consider $(u,x),(v,y) \in \X$ such that $ux^\omega = vy^\omega$.
  Assume without loss of generality that $|u| \geq |v|$, else we just swap them.
  Pick $m = (|u| - |v|) \bmod |y|$ and decompose $y = y_0y_1$ with $|y_0| = m$.
  Then we can write $u = vy^ky_0$ and $x^i = (y_1y_0)^j$ for some $k \geq 0, i, j > 0$.
  It now follows from the properties of $\F$, that 
  \begin{align*}
    (v,y) & \redequiv^\X_\F (vy^{k+1},y) & & \text{loopshift-stable for $\X$ and \cref{lem:refset}~(3)}\\
    & \redequiv^\X_\F  (vy^ky_0y_1, y_0y_1)    &  & \text{$y = y_0y_1$}                 \\
    & \redequiv^\X_\F  (vy^ky_0, y_1y_0)      &  & \text{loopshift-stable for $\X$ and \cref{lem:refset}~(3)} \\
    & \redequiv^\X_\F  (vy^ky_0, (y_1y_0)^j)  &  & \text{power-stable for $\X$ and \cref{lem:refset}~(2)} \\
    & \redequiv^\X_\F  (u, x^i) & & u =vy^ky_0 \text{ and } x^i = (y_1y_0)^j\\
    & \redequiv^\X_\F  (u, x)  &  & \text{power-stable for $\X$ and $(u,x) \in \X$}
  \end{align*}
  Thus, we have $(u,x) \redequiv^\X_\F (v,y)$, and the claim is established.
\end{proof}

\begin{lemma}
    Let $\F = (\T, \D)$ be an "FDFA" of "size@@family" $(n,k)$ and $\X$ be a "reference set".
    It is possible to decide in polynomial time whether $\F$ is "loopshift-stable for $\X$".
    If not, one obtains in the same time a counterexample $(u,ax) \in \X$ of polynomial length such that $(u, ax) \not\redequiv^\X_\F (ua,xa)$.%
  \label{app:lemma:rotationinvariancedecidable}
\end{lemma}
\begin{proof}
  $\F$ is not "loopshift-stable for $\X$" if there are $(u,ax), (ua,xa) \in \X$ with $(u,ax) \not\redequiv^\X_\F (ua,xa)$.
  To verify whether such a counterexample exists, we consider all $u \in \states{\T}$ and $a \in \Sigma$.
  As guaranteed by (5) of~\cref{lem:refset}, we assume that $\C_u$ is a DFA of size at most $|\T|$ accepting all $x \in \Sigma^+$ such that $(u,x) \in \X$.
  We now construct
  \begin{itemize}
    \item $\A_{u,a}$ accepting all words $w \in \Sigma^+$ such that
      $aw \in \LangDFA{\D_u} \cap \LangDFA{\C_u}$
    \item $\B_{u,a}$ which accepts a word $w \in \Sigma^+$ if $wa \in
      \LangDFA{\D_{ua}}\cap \LangDFA{\C_{ua}}$.
  \end{itemize}
  $\A_{u,a}$ can be obtained by first building the product of the two automata for the intersection and then making the $a$-successor of the original initial state as the new initial state.
  Similarly, for $\B_{u,a}$ one first builds the product of the two automata for the intersection, and subsequently marks precisely those states as accepting that have a transition on $a$ into a final state.
  A word $x$ in the symmetric difference of $\LangDFA{\A_{u,a}}$ and $\LangDFA{\B_{u,a}}$ now witnesses that $\A$ is not loopshift-stable for $\X$.
  
  As the constructed DFAs are obtained from the transition system resulting from the product of $\T$ and a progress DFA, their size is clearly polynomial, which means that testing the symmetric difference for emptiness can be done in polynomial time.
  Since we have to consider $|\T|$ choices for $u$ and $|\Sigma|$ choices for $a$, so the overall runtime of the procedure is clearly polynomial.
  Finally, since non-emptiness is witnessed by a word $x$ of linear length, and $|u|$ is at most $|\T|$ (since $\states{\T}$ is prefix-closed), the a counterexample of polynomial length is guaranteed to exist.
\end{proof}

Deciding whether the language accepted by a DFA is "power-stable@@for" is, in general, a difficult problem (which can be shown by a reduction similar to the one in ~\cref{lemma:powerclosednessPSPACEmembership}).
However we can show that if $\F$ is "loopshift-stable@@for", then each "progress DFA" $\D_u$ of $\F$ has a transition structure that is monoidal in the sense made precise by the following lemma.

\begin{lemma}
  Let $\X$ be a "reference set" and $\F = (\T, \D)$ be an "FDFA" that is "loopshift-stable for $\X$" in which each "progress DFA" is minimal.
  For all $(u,x_1), \dots, (u,x_n), (u,y_1),\dots,(u,y_n) \in \X$ with $\D_u(x_i) = \D_u(y_i)$ for $1 \leq i \leq n$, it holds that $\D_u(x_1 \dotsc x_n) = \D_u(y_1 \dotsc y_n)$.%
  \label{app:lemma:rotationinvarianceproperty}
\end{lemma}
\begin{proof}
  Since each right congruences representing the transition system of a progress DFA refines the right congruence of the leading transition system, we can safely assume that $\F$ accepts only representations from $\X$.
  Now let $(u,x_1), \dots, (u,x_n), (u,y_1),\dots,(u,y_n) \in \X$ satisfy the conditions of the statement, i.e.~$\D_u(x_i) = \D_u(y_i)$ for all $1 \leq i \leq n$.
  First, observe that $(u,x_1), (u,y_1) \in \X$ implies $(u,x_1y_1) \in \X$ by (4) of~\cref{lem:refset}.
  More generally, iterated application yields that for any $w$ which is the concatenation of an arbitrary number of $x_i$ and $y_j$ holds that $(u,w) \in \X$.
  Moreover, the same holds if we shift an arbitrary number of these $x_i$ or $y_i$ into the prefix, as guaranteed by property three of~\cref{lem:refset}.
  
  Now let $z \in \Sigma^*$ be such that $(u,x_1\dots x_nz) \in \X$, then
  \begin{align*}
  (u, x_1\dots x_nz) & \redequiv^\X_\F (u,y_1x_2\dots x_nz) & & \D_u(x_1) = \D_u(y_1)\\
  & \redequiv^\X_\F (uy_1,x_2\dots x_nzy_1) && \text{loopshift-stable for $\X$}\\
  & \redequiv^\X_\F (uy_1,x_2\dots x_nzy_1) && \T(uy_1) = \T(u)\\
  & \redequiv^\X_\F (uy_1, y_2x_3\dots x_nzy_1) && \D_u(x_2) = \D_u(y_2)\\
 & \redequiv^\X_\F \dots \\
 & \redequiv^\X_\F (uy_1\dots y_{n-1}, y_nzy_1y_2\dots y_{n-1})\\
 & \redequiv^\X_\F (uy_1\dots y_{n-2}, y_{n-1}y_nzy_1\dots y_{n-2}) \\
 & \redequiv^\X_\F \dots \\
 & \redequiv^\X_\F (u, y_1\dots y_n z).
  \end{align*}
  We can shift the prefix back into the loop in the end, since we know that $(u,y_1\dots y_nz) \in \X$, and therefore by property 3 of~\cref{lem:refset} also $(uy_1\dots y_i,y_{i+1}\dots y_nzy_1\dots y_i) \in \X$ for all $i < n$.
  So for all $z \in \Sigma^*$ holds that $(u,x_1\dots x_nz) \notin \X$ or $x_1\dots x_nz \in \LangDFA{\D_u}$ if and only if $y_1\dots y_nz \in \LangDFA{\D_u}$.
  Since we assumed $\D_u$ to only accept words $w$ such that $(u,w) \in \X$, that must mean the left quotients of $x_1\dots x_n$ and $y_1\dots y_n$ with regard to $\LangDFA{\D_u}$ coincide.
  Since $\D_u$ is minimal by assumption, this immediately gives us $\D_u(x_1\dots x_n) = \D_u(y_1\dots y_n)$.
\end{proof}

A key implication of this technical lemma is that the powers of all words reaching the same state in a "progress DFA" $\D_u$ behave the same.
Specifically, for $x, y \in \Sigma^+$ with $\D_u(x) = \D_u(y)$, it holds that $\D_u(x^i) = \D_u(y^i)$ for all powers $i \geq 1$.
This implies that for checking "power-stability" it suffices to consider one representative for each state of $\D_u$, which results in an efficient decision procedure.

\begin{lemma}
    Let $\X$ be a "reference set" and $\F = (\T, \D)$ be an "FDFA" which is "loopshift-stable for $\X$".
    For any $(u,x) \in \X$ and $y \in \Sigma^*$ with $\D_u(x) = \D_u(y)$, we have that $\D_u(x^i) = \D_u(y^i)$ for all $i > 0$.
    In particular, this implies that $\F$ is "power-stable for $\X$" if and only if for all $u \in \states{\T}, x \in \states{\D_u}$ with $(u,x) \in \X$ holds $(u,x^i) \redequiv^\X_\F (u,x^j)$ for all $0 < i < j \leq |\D_u|$.%
  \label{app:lemma:powerclosednessdecidable}
\end{lemma}
\begin{proof}
    Let us begin with the first claim, which we then use to show the equivalence in the second claim.
    Let $(u,x) \in \X$ and $y \in \Sigma^*$ with $\D_u(x) = \D_u(y)$.
    For $i \geq 1$, we can apply \cref{app:lemma:rotationinvarianceproperty} to $i$-copies of $(u,x)$ and $(u,y)$, giving us that $\D_u(x^i) = \D_u(y^i)$, which establishes the first claim.\\
    For the equivalence from the second claim, we show that if $(u, x^i) \not\redequiv^\X_\F (u,x^j)$ for $0 < i < j$, then already $(\T(u),\D_u(x)^{i'}) \not\redequiv^\X_\F (\T(u),\D_u(x)^{j'})$ for some $0 < i' < j' \leq |\D_u|$.
    Consider a representation $(u,x) \in \X$ witnessing that $\F$ is not "power-stable for $\X$", meaning $(u, x^i) \not\redequiv^\X_\F (u,x^j)$ for $0 < i < j$.
    Let $\hat{u} = \T(u)$, then $\T(u) = \T(\hat{u})$ giving us $(u, x^i) \redequiv^\X_\F (\hat{u}, x^i)$ by (1) of~\cref{lem:refset}.
    Moreover, for $\hat{x} = \D_u(x)$, we have $\D_u(x) = \D_{\hat{u}}(x) = \D_{\hat{u}}(\hat{x}) = \D_u(\hat{x})$.
    Applying the first claim that we already proved, we get that $(u,x^i) \redequiv^\X_\F (\hat{u}, x^i) \redequiv^\X_\F (\hat{u}, \hat{x}^i)$.
    
    To see that $i,j \leq |\D_u|$, consider the sequence $(q_m)_{m \in \mathbb{N}}$ defined by $q_m = \D_u(\hat{x}^m)$.
    By the pigeonhole principle, there must be a repetition among the first $|\D_u| + 1$ states.
    So there exist $0 \leq k < l \leq |\D_u|$ with $q_k = q_l$, and moreover the sequence repeats cyclically afterwards.
    A standard pumping argument shows that we can cut out repetitions of this sequence, and thereby bound the exponents $i,j$ by the size of $\D_u$, concluding the proof.
\end{proof}

If we first check "loopshift-stability" through~\cref{app:lemma:rotationinvariancedecidable}, and subsequently test for "power-stability" through~\cref{app:lemma:powerclosednessdecidable}, then by~\cref{app:lemma:fdfacharacterisation} we obtain a decision procedure for saturation.
\begin{theorem}
    For a "reference set" $\X$ and a given "FDFA" $\F$, one can decide in polynomial time whether $\F$ is "saturated for $\X$".
    If not, there exist $(u,x),(v,y) \in \X$ of polynomial length such that $ux^\omega = vy^\omega$ and $(u,x) \not\redequiv^\X_\F (v,y)$.
  \label{app:theorem:ptimesaturationFDFA}
\end{theorem}
\begin{proof}
    From \cref{app:lemma:fdfacharacterisation}, it follows that $\F = (\T, \D)$ is saturated for the "reference set" $\X$ if and only if $\F$ is both "loopshift-stable for $\X$" and "power-stable for $\X$".
    By \cref{app:lemma:rotationinvariancedecidable}, we know that it is decidable in polynomial time whether $\F$ is loopshift-stable.
    If not, the lemma gives a counterexample of polynomial length witnessing that $\F$ is not saturated for $\X$.
    Otherwise, we may assume that $\F$ is loopshift-stable for $\X$, which allows us to apply \cref{app:lemma:powerclosednessdecidable} to test whether $\F$ is also power-stable for $\X$.
    
    Regarding the length of a counterexample, if the first test fails, we obtain a counterexample $(u, ax) \not\redequiv^\X_\F (ua,xa)$ with $|u| \leq n, |x| \leq n^2k^2$, which satisfies the bounds.
    Similarly, the second test yields a counterexample consisting of $u \in \states{\T}, x\in\states{\D_u}$ and indices $0 <i < j \leq |\D_u|$ such that $(u,x^i) \not\redequiv^\X_\F (u,x^j)$.
    As $|u| \leq |\states{\T}|$ and $|x| \leq |\states{\D_u}|$ guarantees that a counterexample of polynomial length exists.
    
    As we already argued, testing for loopshift-stability is possible in polynomial time.
    If this check returns a positive answer, we consider $n$ choices for $u$ and at most $k$ possibilities for $x$.
    By (5) of~\cref{lem:refset}, verifying if $(u,x) \in \X$ is linear in $|u|$ and $|x|$.
    What remains is to compute the sequence of states reached by $\D_u$ for the first $k + 1$ powers of $x$, and to check whether both a final and a non-final state appear in the sequence.
    This is clearly possible in polynomial time, and the claim is established.
\end{proof}

In contrast to FDFAs, the mode of acceptance by itself ensures that all repetitions of the looping part are treated the same in an FDWA.
Indeed, since $(u,x)$ is in $\ReprFDWA{\W}$ precisely if $x^\omega$ is accepted by the "progress DFA" for $u$, which immediately ensures that no repetition of $x$ can lead to a different acceptance status.
So for deciding whether an FDWA is saturated, what remains is to check is that shifting parts of the loop into the prefix does not change acceptance.

\begin{theorem}
  For a "reference set" $\X$ and a given "FDWA" $\W$, one can decide in polynomial time polynomial whether $\W$ is "saturated for $\X$".
  \label{app:theorem:ptimesaturationFDWA}
\end{theorem}
\begin{proof}
In the latter part of this proof, we show that a witness for $\W = (\T, \B)$ not being saturated corresponds to $u \in \states{\T}, p,q \in \states{\B_u}$ and $r \in \states{\B_v}$ where $v = uv_p$ for some $v_p$ that reaches $p$ in $\D_u$, such that there exist $x, y$ with the following properties:
  \begin{enumerate}
      \item $(u,xy) \in \X$ (whence $(ux,yx) \in \X$ by (3) of~\cref{lem:refset})
      \item $x$ reaches $p$ in $\B_u$ and leads from $q$ to $p$ in $\B_u$
      \item $y$ leads from $p$ to $q$ in $\B_u$ and reaches $r$ in $\B_v$
      \item $xy$ loops on $r$ in $\B_v$.
      \item $p,q$ are accepting in $\B_u$ if and only $r$ is not accepting in $\B_v$.
  \end{enumerate}
  So an algorithm testing for saturation can iterate over all possible choices for the variables and test if such words exist.
  To do this, it first filters out choices for the variables such that (1) cannot be satisfied, which can be done using the DFA guaranteed by (5) of~\cref{lem:refset}.
  Then, it build DFAs $\A$ for words $x$ satisfying (2), and $\B$ for words $y$ satisfying (3).
  Finally, we can construct an NFA $\C$ as the product between a DFA for words looping on $r$ in $\B_v$, and the concatenation of $\A$ and $\B$.
  Now the words accepted by $\C$ correspond precisely to witnesses of non-saturation of $\W$.

  Filtering out choices for the variables that cannot satisfy (1) is easy and can be done in time $nk$ by testing if the DFA for words looping on $u$ accepts some $x$ that reaches $p$ in $\B_u$.
  Similarly, property (5) can easily be checked.
  The DFAs $\A$ and $\B$ have at most $k^2$ states, so $|\C| \leq k^5$.
  As we iterate over $n$ choices for $u$ and at most $k$ choices for $p,q,r$ each, the overall runtime of the algorithm is at most $nk^8$, which is clearly polynomial in $n$ and $k$.
  Finally, we can bound the maximum length of a counterexample for $\W$ being saturated by $n$ and the size of $\C$.
  
  We now show that counterexamples to the saturation of $\W = (\T, \B)$ are indeed of the claimed form, which concludes the proof.
  So let $(u,x) \in \X$, then every state of $\B_u$, either loops on some $x^i$ or it reaches on some $x^j$ a state that loops on $x^i$.
  Since $\B_u$ has at most $k$ states, it must be that $j \leq k$.
  If we pick $e = k!$, then $x^e$ reaches some state $q$ and $x^e$ also loops on $q$ since every possible (elementary) cycle length is at most $|\B_u$ and divides $e$.

  Consider $(u,x),(v,y) \in \X$ with $ux^\omega = vy^\omega$.
  By the previous argument, $x^e$ reaches and loops on some $p$ in $\B_u$, while $y^e$ reaches and loops on some $r$ in $\B_v$.
  We revisit the proof of~\cref{app:lemma:fdfacharacterisation} to find a decomposition $u = vy^my_0$ and $x^i = (y_1y_0)^j$ for $m \geq 0, i,j > 0$.
  Since $x^e$ ($y^e$) reaches and loops on $p$ ($r$) in $\B_u$ ($\B_v$), this holds true for any multiple, so in particular $x^{e\cdot i}$ ($y^{e\cdot j}$).
  For $q = \delta^*_u(p, y_1)$, we now have the following situation:
  \begin{center}
      \begin{tikzpicture}[node distance=18mm]
          \node[] (ue) {$\B_u:\; \eps$};
          \node[right=of ue] (p) {$p$};
          \node[right=of p] (q) {$q$};
          \node[right=of q,xshift=.1\textwidth] (ve) {$\B_v:\; \eps$};
          \node[right=of ve] (r) {$r$};
          
          \draw[->]
          (ue) edge node[anchor=south] {$x^{e\cdot i}$} (p)
          (p) edge node[anchor=south] {$y_1$} (q)
          (q) edge[loop right] node {$(y_0y_1)^{e\cdot j}$} (q)
          (q) edge[bend left=15] node[anchor=north] {$(y_0y_1)^{e \cdot (j-1)}y_0$} (p);
          \draw[->]
          (ve) edge node[anchor=south] {$y^{e\cdot j}$} (r)
          (r) edge[loop right] node {$y^{e\cdot j}$} (r);
      \end{tikzpicture}
  \end{center}
  We now define
  \begin{align*}
      \hat{x} &= x^{e\cdot i}y_1 = (y_1y_0)^{e\cdot j}y_1 = y_1(y_0y_1)^{e\cdot j} = y_1y^{e\cdot j} \quad \text{ and}\\
      \hat{y} &= y^{e\cdot j}y_0 = (y_0y_1)^{e\cdot j}y_0 = y_0(y_1y_0)^{e\cdot j} = y_0x^{e\cdot i}.
  \end{align*}
  Our goal is now to show that $(u,x)\redequiv^\X_\W (u,\hat{x}\hat{y})$ and $(v,y)\redequiv^\X_\W (uy_1,\hat{y}\hat{x})$, which means an algorithm testing for saturation only has to consider witnesses of this kind.
  By simple rearrangements, we obtain the following equalities.
  \begin{alignat*}{2}
          u\cdot \hat{x}\cdot\hat{y} &= u \cdot y_1y^{e\cdot j} \cdot y_0x^{e\cdot (i-1)} & 
          \qquad uy_1 \cdot\hat{y}\cdot\hat{x} &= uy_1\cdot (y_0y_1)^{e\cdot(j-1)} y_0\cdot (y_1y_0)^{e\cdot j}y_1\\
          &= u\cdot y_1(y_0y_1)^{e\cdot j} y_0(y_1y_0)^{e\cdot(j-1)} &&= uy_1\cdot (y_0y_1)^{e\cdot j - 1 + e \cdot j + 1}\\
          &= u\cdot (y_1y_0)^{e\cdot j + 1 + e \cdot j - 1} &&= uy_1\cdot y^{2e\cdot j}\\
          &= u\cdot (y_1y_0)^{2\cdot e \cdot j}\\
          &= u\cdot x^{2e\cdot i}
  \end{alignat*}
  Since $(u,x) \in \X$, we obtain by (2) of~\cref{lem:refset} that also $(u,\hat{x}\hat{y}) = (u,x^{2e\cdot i}) \in \X$.
  Similarly, it follows from (2) and (3) of the same Lemma that $(uy_1, \hat{y}\hat{x}) = (uy_1, y^{2e\cdot j}) \in \X$.
  Moreover, since exponentiation of the loop naturally preserves acceptance in an FDWA, we obtain the desired $\redequiv^\X_\F$-equivalences.
  Finally, note that (1) of~\cref{lem:refset} guarantees that we may assume $|u| \leq n$ as we could replace it with a shorter word reaching the same state in $\T$.
\end{proof}

  \subsection{Full proofs of~\cref{sec:learning}}\label{section:appendixlearning}

Formally, an active learner gets two functions as input, an $L$-membership oracle and an $L$-equivalence oracle for a target language $L \in \L$. For a regular language $L$, an $L$-membership oracle is a function $\moracle:\Sigma^* \rightarrow \{0,1\}$ such that $\moracle(u) = 0$ iff $u \notin L$. And an $L$-equivalence oracle for DFAs is a function $\eoracle(\A)$ that takes DFAs as input, and returns $\top$ if $\LangDFA\A = L$, and otherwise returns a counter-example $u \in \Sigma^*$ such that $u \in L$ iff $u \notin \LangDFA\A$. Similarly, for a regular $\omega$-language $L$, an \AP $L$-""membership oracle"" is a function $\moracle: \UP{\Sigma} \rightarrow \{0,1\}$ such that $\moracle(u,v) = 0$ iff $uv^\omega \notin L$. An \AP $L$-""equivalence oracle"" for fully saturated FDFAs is a function $\eoracle(\F)$ that takes fully saturated FDFAs as input, and returns $\top$ if $\LangFDFA\F = L$, and otherwise returns a counter-example $(u,v) \in \UP{\Sigma}$ such that $uv^\omega \in L$ iff $uv^\omega \notin \LangFDFA\F$.

\activelearner*
\begin{proof}
  Our algorithm uses a polynomial time active learner for DFAs in order to learn a DFA for the language $L_\$ := \{u\$v \mid uv^\omega \in L\}$ over the alphabet $\Sigma \cup \{\$\}$ (as mentioned in the related work, this approach is also used in \cite{FarzanCCTW08} for building an NBA learner). Any DFA $\A$ over $\Sigma \cup \{\$\}$ can easily be turned into an FDFA $\F_\A$ with $\ReprFDFA{\F_\A} = \{(u,v) \in \Sigma^* \times \Sigma^+ \mid u\$v \in \LangDFA\A\}$ such that the leading transition system and the progress DFAs of $\F_\A$ all have size at most $|\A|$: the leading transition system $\T$ is the transition system of $\A$ restricted to $\Sigma$ and to the states reachable without the special letter $\$$. And for each $q$ in $\T$, the progress DFA $\D(q)$ is $\A$ with the $\$$-successor of $q$ as initial state and the alphabet restricted to $\Sigma$. 
  
  Vice versa, any FDFA $\F$ can be turned into a DFA $\A_\F$ with $\LangDFA{\A_\F} = \{u\$v \mid (u,v) \in \ReprFDFA\F\}$ by just taking the disjoint union the leading transition system and the progress DFAs, and inserting $\$$-transitions from each $q$ of the leading transition system to the initial state of its progress DFA.  This implies that the minimal DFA for $L_\$$ is as most as big as the minimal fully saturated FDFA for $L$.

  We now describe our "active learner" for fully saturated FDFAs.   We have to build an algorithm that gets as input for some regular $\omega$-language $L \subseteq \Sigma^\omega$
  \begin{itemize}
  \item an $L$-"membership oracle" $\moracle(u,v)$ for ultimately periodic words, and 
  \item an $L$-"equivalence oracle" $\eoracle(\F)$ for fully saturated FDFAs.
  \end{itemize}
Our algorithm runs a polynomial time active learning algorithm for DFAs over the alphabet $\Sigma \cup \{\$\}$. For this it implements an $L_\$$-membership oracle $\moracle_\$(x)$ for finite words, and an $L_\$$-equivalence oracle $\eoracle_\$(\A)$ for DFAs as follows:
  \begin{itemize}
  \item $\moracle_\$(x)$ returns $\moracle(u,v)$ if $x= u\$v$ with $u \in \Sigma^*$ and $v \in \Sigma^+$, and otherwise $\moracle_\$(x)$ returns $0$.
  \item $\eoracle_\$(\A)$ first checks if $\F_\A$ that is constructed as described above, is fully saturated (\cref{theorem:ptimesaturationFDFA}).
    \begin{itemize}
    \item If $\F_\A$ is not fully saturated, then the algorithm for the full saturation test returns counter-examples $(u_1,v_1)$, $(u_2,v_2)$ with $u_1v_1^\omega = u_2v_2^\omega$ but $(u_1,v_1) \in \ReprFDFA{\F_\A}$ and $(u_2,v_2) \notin \ReprFDFA{\F_\A}$. Our function $\eoracle_\$(\A)$ returns $u_1\$v_1$ if $\moracle(u_1,v_1) = 0$, and $u_2\$v_2$ otherwise. 
    \item If $\F_\A$ is fully saturated, then $\eoracle_\$(\A)$ checks if $\eoracle(\F_\A)$ returns $\top$. If yes,  $\eoracle_\$(\A)$ also returns $\top$. Otherwise, $\eoracle(\F_\A)$ has returned some pair $(u,v)$ and $\eoracle_\$(\A)$ returns $u\$v$.
  \end{itemize}
  \end{itemize}
  Note that the equivalence oracle $\eoracle_\$(\A)$ does not care if $\A$ accepts any words with no $\$$ or with more than one $\$$. Such words are filtered out in the construction of $\F_\A$.
  From the definition of the oracles $\moracle_\$(u)$ and $\eoracle_\$(\A)$, it follows that these are indeed oracles for $L_\$$, as long as $\F_\A$ does not accept the target language. 
    When the DFA learner terminates because the last call $\eoracle_\$(\A)$ returned $\top$, then we know that $\eoracle(\F_\A)$ has returned $\top$, so our FDFA learner can output $\F_\A$.

    It remains to argue that the algorithm runs in time polynomial in the size of the minimal fully saturated FDFA $\F_L$ for the $\omega$-language $L$, and in the length of the longest counter-example returned by $\eoracle(\F)$ during the run of the algorithm.

  The running time of the active DFA learner is polynomial in the size of the minimal DFA $\A_{L_\$}$ for $L_\$$, and in the length of the longest counter-examples returned by $\eoracle_\$(\A)$.
  From the translation between DFAs and FDFAs described at the beginning of the proof we get that $\A_{L_\$}$ has size at most $|\F_L|$. 
  The counter-examples returned by $\eoracle_\$(\A)$ are either counter-examples $u\$v$ where $(u,v)$ was a counter-example returned by $\eoracle(\F_\A)$, or counter-examples resulting from the full saturation check on $\F_\A$. The latter counter-examples are polynomial in the size of $\F_\A$ by \cref{theorem:ptimesaturationFDFA}, and hence polynomial in $\F_L$ because the DFAs used in $\eoracle(\A)$ by the active DFA learner are polynomial in the size of $\A_{L_\$}$. Furthermore, the running time of $\eoracle_\$(\A)$ is polynomial in $\A$ because the full saturation check is polynomial by \cref{theorem:ptimesaturationFDFA}.

  Hence, the execution of the active DFA learner takes time polynomial in $\F_L$ and in the length of the longest counter-example returned by $\eoracle(\F)$.
\end{proof}

\passivelearner*
\begin{proof}
  The "passive learner" works as follows:
  \begin{enumerate}
  \item Construct an FDFA $\F$ from the given sample $S$ using techniques known from passive learning of DFAs.
  \item  Check whether $\F$ is a syntactic FDFA that is consistent with $S$ and return $\F$ if yes. 
  \item Otherwise, return a default FDFA that accepts precisely all representations of the ultimately periodic words in $S_+$.
  \end{enumerate}
  In Step~2 of this algorithm, checking whether $\F$ is consistent with $S$ can be done by simply running all the examples on $\F$. The test whether $\F$ is a syntactic FDFA can then be done by checking whether $\F$ is saturated. We guarantee in step 1 that there is no smaller leading congruence for the given sample.  Hence, if $\F$ is saturated and consistent with $S$, then it is a syntactic FDFA.

  Concerning step 3, constructing a syntactic FDFA for precisely the ultimately periodic words in $S_+$ is straight-forward, so we do not detail the construction here.

  It remains to describe step 1, for which it is in principle known how to do it. For example, the generic algorithm GleRC from \cite{BohnL24} for inferring right-congruences from samples of finite words can be used for this. To be self-contained, we describe a possible method in the following, together with an argument that each syntactic FDFA can be learned in the limit from polynomial data.
  
  Let $L \subseteq \Sigma^\omega$ be a regular $\omega$-language. The leading congruence of the syntactic FDFA for $L$ is the transition system of the canonical right congruence $\sim_L$ over $\Sigma^*$, defined by
  \[
  u \sim_L u' :\Leftrightarrow (\forall w \in \Sigma^\omega:\; uw \in L \Leftrightarrow u'w \in L).
  \]
We say that an $L$-sample $\sim_L$-separates two words $u,u' \in \Sigma^*$ if the sample contains examples $(uv,x)$ and $(u'v,x)$ such that $uvx^\omega \in L$ iff $u'vx^\omega \notin L$. 
  
Let $R_{\sim_L} = \{u_1, \ldots, u_n\}$ be the set of length-lexicographic (llex) least representatives of the $\sim_L$-classes. Then there is an $L$-sample $S_{\sim_L}$ of size polynomial in the number of classes of $\sim_L$ that contains examples that $\sim_L$-separate all $u_i,u_j$ with $i \not=j \in \{1, \ldots, n\}$, and all $u_ia,u_j$ with $u_ia \not\sim_L u_j$ for all $i,j \in \{1, \ldots, n\}$ and $a \in \Sigma$. The following algorithm constructs the transition system of $\sim_L$ for each $L$-sample $S$ that contains all examples from $S_{\sim_L}$:
\begin{itemize}
\item Start with $R := \{\varepsilon\}$ and while there exists $v \in \Sigma^*$ such that $v$ is $\sim_L$-separated by $S$ from all $u \in R$, add the llex-smallest such $u$ to $R$.
\item Use the elements of $R$ as states, and for each $u \in R$ and $a \in \Sigma$, add a $a$-transition from $u$ to $u'$ for the least $u' \in R$ such that $ua$ and $u'$ are not $\sim_L$-separated by $S$.
\end{itemize}
This finishes the construction of the leading transition system. Note that the algorithm guarantees that there cannot be a smaller leading congruence for the given sample because all the elements of $R$ are separated by examples.

For each $u \in R_{\sim_L}$, the progress DFA is the minimal DFA accepting those $x \in \Sigma^+$ with $u \sim_L ux$ and $ux^\omega \in L$. The transition system of this DFA corresponds to the right-congruence $\approx_L^u$ defined by
\[
x \approx_L^u y \; :\Leftrightarrow \; x \sim_L y
\text{ and } \forall z \in \Sigma^*: u \sim_L uxz \Rightarrow (u(xz)^\omega \in L \Leftrightarrow u(yz)^\omega \in L).
\]
The final states are the classes of those $x$ with $u \sim_L ux$ and $ux^\omega \in L$. 
We use the convention that $ux^\omega \notin L$ if $x$ is empty. This corresponds to the initial state of the progress DFAs being non-accepting.

Now we can proceed in the same way as for $\sim_L$. We say that an $L$-sample $\approx_L^u$-separates $x,y \in \Sigma^*$ if the sample contains examples $(u,xz)$ and $(u,yz)$ such that $u(xz)^\omega \in L$ iff $u(yz)^\omega \notin L$. 
Let $u \in R_{\sim_L}$ and let $P_{\approx_L^u} = \{x_1, \ldots, x_m\}$ be the set of the llex-least representatives of the $\approx_L^u$-classes. Then there is an  $L$-sample $S_{\approx_L^u}$ of size polynomial in the number of classes of $\approx_L^u$ that contains examples that $\approx_L^u$-separate all $x_i,x_j$ with $i \not=j \in \{1, \ldots, m\}$, and all $x_ia,x_j$ with $x_ia \not\approx_L^u x_j$ for all $i,j \in \{1, \ldots, m\}$ and $a \in \Sigma$. For defining the accepting states of the progress DFA,  $S_{\approx_L^u}$ also contains the examples $(u,x_i)$ for all $i \in \{1,\ldots,m\}$ such that $ux_i \sim_L u$ and $u(x_i)^\omega \in L$.
Now the progress DFA for $u$ can be inferred by the same type of algorithm as for the leading transition system:
\begin{itemize}
\item Start with $P := \{\varepsilon\}$ and while there exists $x \in \Sigma^*$ such that $x$ is $\approx_L^u$-separated by $S$ from all $x \in P$, add the llex-smallest such $x$ to $P$.
\item Use the elements of $P$ as states, and for each $x \in P$ and $a \in \Sigma$, add a $a$-transition from $x$ to $y$ for the least $y \in P$ such that $xa$ and $y$ are not $\approx_L^u$-separated by $S$.
\end{itemize}
For each state $u$ of the leading transition system, this algorithm infers the progress DFA for $u$ for each $L$-sample that contains all examples of $S_{\approx_L^u}$.
The characteristic sample for the syntactic FDFA is obtained by taking $S_{\sim_L}$ and all the $S_{\approx_L^u}$.
\end{proof}

  \subsection{Full proofs of~\cref{section:almostsaturation}}\label{section:appendixalmostsaturation}
In this section, we provide full proofs for showing that deciding almost saturation is \PSPACE-complete.

\almostsaturationPSPACEcomplete*

We show membership in \PSPACE first by using standard techniques.

\begin{lemma}
  Deciding whether a given "FDFA" $\F = (\T, \D)$ is "almost saturated" is in \PSPACE.%
  \label{lemma:powerclosednessPSPACEmembership}
\end{lemma}
\begin{proof}
A counterexample to the "almost saturation" of $\F$ is $(u,x) \in \NormFDFA{\F}$ such that $(u,x^i) \notin \NormFDFA{\F}$ for some $i > 1$.
    We can nondeterministically guess such a word $x = a_1a_2\dotsc a_n$ letter by letter.
    If $x_k = a_1\dots a_k$ is the word guessed so far, we also keep track of the transformation $x_k$ induces on $Q$, which is a mapping $m_k : Q \to Q$ such that $q \mapsto \delta^*(q, x_i)$.
    As these mappings compose naturally, we can compute them in polynomial space alongside the guessed word.
    It is then easy to obtain $q_i = m_k^i(\iota)$ and $q_j = m_k^j(\iota)$, the states reached by $x_k^i$ and $x_k^j$, respectively.
    Since for each guessed letter, we can easily check whether $q_j$ is accepting while $q_j$ is not, \PSPACE membership is established.
\end{proof}

For \PSPACE-hardness, we reduce from the DFA intersection problem.
An instance of this problem is a sequence $\mathbb{S} = (\D_1, \dots,\D_p)$ of DFAs, and the goal is to decide whether the intersection of all $\D_i$ is non-empty.
It is well-known that deciding whether some word is accepted by all $\D_i$ is \PSPACE-complete~\cite{Kozen77}.

\begin{lemma}
    Deciding if a given "FDFA" $\F = (\T, \D)$ is "almost saturated" is \PSPACE-hard.
    \label{lemma:powerclosednessPSPACEhard}
\end{lemma}
\begin{proof}
  We give a reduction from the DFA intersection problem, so let $\mathbb{S} = (\D_1,\dotsc,\D_p)$ be a sequence of DFAs $\D_i = (\Sigma, Q_i, \delta_i, \iota_i, F_i)$.
  In the following, we need that $p$ is prime and greater than $1$.
  This is always possible since by the Bertrand-Chebyshev theorem, there must exist some $p < p' < 2p$ which is prime, and we may simply pad the sequence with universal DFAs.
  The padded sequence has a non-empty intersection if and only if the original sequence has a non-empty intersection since $\Sigma^*$ is neutral with regard to intersection.
  Therefore, we may assume $p$ to be prime and greater than $1$.
  Let $\#$ be a fresh symbol that does not occur in $\Sigma$.
  Define the extended alphabet $\Sigma_\# = \Sigma \mathbin{\dot\cup} \{\#\}$ and consider the language $L \subseteq \Sigma_\#^*$ 
  \begin{align*}
    L = \Sigma_\#^+ \setminus \big(\# \LangDFA{\D_1} \# \LangDFA{\D_2} \# \dotsi \# \LangDFA{\D_p}\big)
  \end{align*}
  We now take the "FDFA" $\F = (\T, \A)$ where $\T$ is trivial and loops on all symbols, and $\A$ is a DFA for $L$ of size at most $2 + \sum_{i \in \{1,\dots,p\}} |\D_i|$ as explained in the following.
  We construct $\A$ by using a new initial state $\iota$, a fresh accepting sink state, and chaining the $\D_i$ as follows.
  From $\iota$, we reach the initial state of $\D_1$ on $\#$ and the accepting sink on all other symbols.
  When reading $\#$ from a state $q$ of a DFA $\D_i$, we go to the new accepting sink if $q$ is not accepting or $i=p$, and to the initial state of $\D_{i+1}$ otherwise.
  A state of $\A$ is accepting if it is not a final state of $\D_p$. This DFA clearly accepts $L$.

 We show that $\F$ is "almost saturated" if and only if $\mathbb{S}$ has empty intersection.
 If $\mathbb{S}$ has non-empty intersection, pick a word $u \in \LangDFA{\D_1}\cap \dots \cap \LangDFA{\D_p}$ and consider $x = \#u$.
 As $\T$ is trivial, $(\eps, x^i)$ is "normalized" for every $i > 0$.
 Since $p > 1$, $x$ does not reach a final state of $\D_p$ in $\A$ and so $x \in \LangDFA{\A}$.
 $x^p$ on the other hand is not accepted by $\A$ as it does reach a final state of $\D_p$ in $\A$.
 Thus, $(\eps,x) \in \NormFDFA{\F}$, but $(\eps, x^p) \notin \NormFDFA{\F}$, witnessing that $\F$ is not "almost saturated".

 For the other direction, assume that $\F$ is not almost saturated, and let $x \in \Sigma_\#^+$ be a witness of this fact, i.e. $x \in \LangDFA{\A}$ but $x^i \notin \LangDFA{\A}$ for some $i > 1$.

  By construction of $\A$, $x^i$ is of the form $x^i = \#u_1\#u_2 \dots \# u_m$ with $m\cdot i=p$ and each $u_j \in \LangDFA{\D_{mk+j}}$ for all $k\ge 0$ such that $mk+j \le p$. Since $i > 1$ and $p$ is prime, we get $m=1$, and thus $x=\#u_1$ with $u_1 \in \LangDFA{\D_j}$ for all $j\in\{1,\ldots,p\}$.
  This then entails that $\mathbb{S}$ has non-empty intersection.
\end{proof}
  \section{Proofs from Section \ref{section:regularity} Regularity}
\label{section:appregularity}

\rstFDWAregular*
\begin{proof}
    The construction that goes back to \cite{CalbrixNP93} and is used in \cite{AngluinBF18} for saturated FDFAs and in \cite{LiST23} for almost saturated FDFAs builds an NBA that accepts an $\omega$-word if it is of the form $uv_1v_2v_3\cdots$ such that
    there is an accepting state $p$ of the progress automaton for $u$, such that each $v_i$ reaches $p$ and loops on $p$. This construction yields an NBA that accepts precisely those ultimately periodic words $uv^\omega$ such that $(u,v)$ is accepted by the FDFA and $v$ loops on the progress state that it reaches. By the semantics of of FDWAs, $uv^\omega$ is in the UP-language if $v^\omega$ is accepted by the progress automaton for $u$. Since this progress automaton is weak (each SCC is either completely accepting or completely rejecting), this means that $v^\omega$ eventually loops on an accepting state. So there must be an $i$ such that $v^i$ loops on the accepting progress state that it reaches. This implies that the constructed NBA contains precisely the UP-words that are in the UP-language of the FDWA, and hence the UP-language of the FDWA is UP-regular.
\end{proof}

\rststabilize*
\begin{proof}
    $\F'$ has to accepts a pair $(u,v)$ if, and only if, there are $v_1,v_2 \in \Sigma^*$ with $v=v_1v_2$ such that $(uv_1,v_2v_1)$ is accepted by $\F$.

    We first observe that $\T(u) = \T(uv)$ implies $\T(uv_1) = \T(uv_1 v_2v_1)$.

    We now discuss how to build $\N_q'$ for $q=\T(u)$.
    To build $\N_q'$, we construct an NFA that accepts $v_1v_2$ if $v_1,v_2 \in \Sigma^*$ and there is a state $p \in \N_{\T(uv_1)}$ s.t.\ there is a path in the nondeterministic transition system of $\N_{\T(uv_1)}$ from an initial state to $p$ when reading $v_2$ \emph{and} from state $p$ to a final state when reading $v_1$.
    For this, $\N_q'$ can simply guess the state $q' = \T(uv_1)$ and the state $p$ of $\N_{q'}$.
    For each such pair, we build an automaton $\N_q^{q',p}$. We describe $\N_q^{q',p}$ as an $\varepsilon$-automaton for convenience.

    $\N_q^{q',p}$ contains a product of $\T$ and two copies of $\N_{q'}$.
    The product of $\T$ is used to keep track of  which state $\T$ is in. So the states of $\N_q^{q',p}$ are of the form $(r,s,i)$ where $r$ is a state of $\T$, $s$ is a state of $\N_{q'}$, and $i \in \{1,2\}$ determines the copy. Within the two copies we have the standard transitions of the product. The initial state is $(q,p,1)$, and the only accepting state is $(q,p,2)$. Additionally, there are $\varepsilon$-transitions from states $(q',p',1)$ to $(q',p_0,2)$ if $p'$ is an accepting state of $\N_{q'}$, and $p_0$ is the initial state of $\N_{q'}$.


    $\N_q'$ is then simply the disjunction over all $\N_q^{q',p}$.

    If our $\F'$ accepts $(u,v)$, we can infer from the accepting run a partition of $v$ into $v_1,v_2$ with $v=v_1v_2$ (the $\varepsilon$-transition is used after reading $v_1$) \emph{and} an accepting run of $\F$ on $(uv_1,v_2v_1)$.
    If we have $\T(u) = \T(uv_1v_2)$, then an accepting run of $\F$ on $(uv_1,v_2v_1)$ provides an accepting run of $\F'$ on $(u,v)$.
\end{proof}

\textbf{Proof of \cref{thm:regular} for that (1) is in \PSPACE}

\begin{proof}
    The idea is to show that (1) is in \PSPACE if (3) is.
    For case (1), we can first use \cref{lem:stabalize}, and thus assume w.l.o.g.\ that $\F=(\T,\N)$ is loopshift-stable.
    We can then observe that $\T$ is deterministic and thus defines a labeling of the words interpreted by $\F$, where the letters of the word are pairs $(q,a) \in \states{\T}\times \Sigma =: \Pi$.

    We first transform an unlabeled word $(u_2,v_2)\in \Sigma^* \times \Sigma$ into a properly labeled word $(u,v) \in \Pi^*\times \Pi^+$ such that $(u_2,v_2)$ is the second projection of $(u,v)$, and for the first projection $(u_1,v_1)$ we have that $u_1\cdot \T(u_2)$ and $u_1\cdot v_1\cdot \T(u_2\cdot v_2)$ are the runs of $\T$ on $u_1$ and $u_1\cdot v_1$, respectively.

    We then build an FDFA $\F'=(\T',\N')$ that accepts a pair $(u,v)$ if, and only if, it is properly labeled and the second projection $(u_2,v_2)\in \NormFDFA{\F}$ (which implies $\T(u_2)=\T(u_2\cdot v_2)$), and where $\T'$ is like $\T$, except that it moves to a fresh sink $\bot$ when reading a letter $(q,a)$ s.t.\ $q$ is not the state of $\T'$. $\N_\bot'$ recognizes the empty language.
    The "UP-language" of $\F'$ is "UP-regular" if, and only if, the "UP-language" of $\F$ is "UP-regular".

    $\F'$ also defines the same "UP-language" as $\F''=(\T',\N'')$, where $\N_\bot''=\N_\bot'$ and all other automata are replaced by the same automaton $\N_\cup$ that recognizes the union $\bigcup_{q\in \states{\T}}\LangDFA{\N_q'}$.

    If $\T_\varepsilon$ denotes the trivial TS, then the "UP-language" of $\F''$ is "UP-regular" if, and only if, the "UP-language" of $(\T_\varepsilon,\N'')$ is "UP-regular".

\end{proof}

\rstregularEquivalent*
\begin{proof}
    If $L$ is UP-regular, then $\synL$ is of finite index by \cite{Arnold85} since $\synL$ is the syntactic congruence introduced in \cite{Arnold85} in the "prefix-independent" case. The other direction can be shown by using \cite[Theorem~5.7]{AngluinBF18} (which itself is based on \cite[Lemma~5]{CalbrixNP93}): If $\synL$ is of finite index, then we can build an FDFA $(\T,\D)$ with trivial leading transition system and progress DFA that corresponds to the (finite)  transition system defined by $\synL$, with a state being accepting if it corresponds to the class of a word $x$ with $x^\omega \in \Pow{P}$. Then $\F$ accepts a pair $(u,x)$ iff $ux^\omega \in L$, and hence it is saturated. By \cite[Theorem~5.7]{AngluinBF18} there is an NBA that accepts a language $L'$ such that $\UP{\Sigma} \cap L' = L$. Hence $L$ is "UP-regular".
\end{proof}

\rstfinfin*

Before turning to the full proof of \cref{lem:fin2fin}, we introduce a useful lemma for constructing roots in general and for constructing rejecting roots in particular.

\begin{restatable}{rstlemma}{rstdifferent}
    \label{lem:different}
    If $x,y$ are different words of equal length $\ell=|x|=|y|$ and $p>2\ell$ is a prime number, then $x \cdot y^{p-1}$ is a root.

    Moreover, if $\tau$ is a "rejecting" transition profile and both $\tpN(x)$ and $\tpN(y)$ are powers of $\tau$, then $\tpN(x\cdot y^{p-1})$ is "rejecting" and $x\cdot y^{p-1} \notin \Pow{P}$.
\end{restatable}


\begin{proof}
    The length of $w = x\cdot y^{p-1}$ is $\ell \cdot p$, and $p$ is not a divisor of $\ell$ (as $p>\ell$ holds).

    Assume for contradiction that $u$ is a shorter root of length $r=|u| < \ell \cdot p$ of  $w$.
    We first consider the case that $r$ is relative prime to $p$. Then $r$ divides $\ell$, and we have that $u^{\ell/r} = x$ due to the first $\ell/r$ repetitions of $u$, and $u^{\ell/r} = y$ due to the second $\ell/r$ repetitions of $u$, contradicting with the fact that $x \neq y$.

    We now look at the case that $p$ divides $r < p\cdot \ell$.
    But then the first $\ell$ letters of $w$ are $x$ \emph{and} equal to $w_{r+1},\ldots,w_{r+\ell}$, while the second $\ell$ letters of $w$ are $y$ \emph{and} equal to $w_{r+\ell+1},\ldots,w_{r+2\ell}$, while the latter is equal to $w_{r+1},\ldots,w_{r+\ell}$.

    We thus get $x=y$, leading to contradiction.
    This provides that $w$ is a root.

    For the second claim, if $\tpN(x)$ and $\tpN(y)$ are the $j$-th and $k$-th power of $\tau$, respectively, then $\tpN(w)$ is the $(j+(p-1)k)$-th power of $\tau$, and as such "rejecting". Moreover, $w^{\omega}\notin \Pow{P}$ as $\tpN(w)$ is "rejecting" and $w$ a root.
\end{proof}

With this lemma in place, we now turn to the proof of \cref{lem:fin2fin}

\begin{proof}[Proof of \cref{lem:fin2fin}.]
    We have defined $\Ter(P)$ such that $x^{\omega}\in \Pow{P}$ if $\tpN(x)$ is "accepting" \emph{or} $x^\omega = u^\omega$ holds for some root $u$ of a word in $\Ter(P)$.

    We thus have $x \synN y$ if $\tpN(x) = \tpN(y)$ holds \emph{and}, for all $z \in \Sigma^*$, $\tpN(xz)$ is "accepting" \emph{or} there is a root $u\in \Ter(P)$ such that $(xz)^\omega = u^\omega$ iff there is a root $v \in \Ter(P)$ such that $(yz)^\omega = v^\omega$.

    If the set of roots in $\Ter(P)$ is finite, then this is a regular property; consequently, $\synN$ (and thus $\synL$) have finite index.

    We now show that, if $\Ter(P)$ contains infinitely many roots, then, for any natural number $n$, the index of $\synN$ is at least $n$---and thus that $\synN$ has infinite index.

    For such an $n$, we pick $n$ different roots $u_1,\ldots,u_n \in \Ter(P)$ that define the same transition profile $\tau'$.

    Let $v_1,\ldots,v_n$ be powers of $u_1,\ldots,u_n$ such that $|v_1|=\ldots=|v_n|$; e.g., this length could be the smallest common multiplier of the lengths of $u_1,\ldots,u_n$.

    Let $i$ be a rejecting power of $\tau'$ and $\tau$ the $i$-th power of $\tau'$ and thus "rejecting", Let $\ell = i \cdot |v_1|$ and $p>2\ell$ a prime number.
    We now have that, for all $j,k\leq n$, ${v_j}^{i}{v_k}^{i(p-1)} = (v_j^i)\cdot (v_k^i)^{p-1}$ is in $\Pow{P}$ if $j=k$ (because its root $u_j$ has an "accepting" transition profile), but it is "rejecting" for $j\neq k$ by Lemma \ref{lem:different}.
\end{proof}

\rstInfiniteGoodness*

We first provide an alternative definition for a "good witness" and show that it is equivalent.

\begin{enumerate}
    \item there are infinitely many words $x$ that visit $\tau_g$ first, or
    \item[2a.] there are two different
          words $x\neq y$ that visit $\tau_g$ first and there is a word $u$ that recurs on $\tau_g$, or
    \item[2b.] there is a word $x$ that visits $\tau_g$ first and there are two different
          words $u \neq v$ that recur on $\tau_g$, or
    \item[2c.] there is exactly one a word $x$ that visits $\tau_g$ first and exactly one word $u$ that recurs on $\tau_g$, and they have different roots.
\end{enumerate}

\begin{lemma}\label{lem:sameGoodness}
    The alternative definition is equivalent to the definition in the paper.
\end{lemma}

\begin{proof}
    The first case is the same, so it is enough to show that the second case from the definition, ``2. there is a word $x$ that visits $\tau_g$ first and a word $u$ that recurs on $\tau_g$ that have different roots'' is equivalent to ``(2a) or (2b) or (2c)''.

    ``$\Rightarrow$'' is trivial.

    ``$\Leftarrow$:'' (2c) implies (2), as (2c) is a special case of (2).
    (2b) implies (2) as $u,v$ are different and recur, so they have different roots. They therefore cannot both have the same root as $x$.
    Similarly, (2a) implies (2) as $x,y$ are different and visit $\tau_g$ first, so they have different roots. They therefore cannot both have the same root as $u$.
\end{proof}

We use the new definition in the long version of the proof.

We split the proof into the following two lemmas.

We first show that, if $\Ter(P)$ contains infinitely many roots, then every NFA $\N$ with $P = \LangDFA{\N}$ has a good witness $\tau_g$.

\begin{lemma}
    If $\Ter(P)$ has finitely many roots, then there is a good witnesses $\tau_g$.
\end{lemma}

\begin{proof}
    Assume that $\Ter(P)$ has infinitely many roots.
    Then, as in the proof of Lemma \ref{lem:fin2fin}, we choose $\tau_g$ as a transition profile that has infinitely many roots $u$ in $\Ter(P)$ such that $\tpN(u) = \tau_g$. We show that $\tau_g$ is a good witness.
    Since there are infinitely many roots $u$ in $\Ter(P)$ such that $\tpN(u) = \tau_g$, we can conclude that there are either infinitely many words that visit $\tau_g$ first, or there is at least one word that visits $\tau_g$ first and infinitely many words loop on $\tau_g$.

    So the only possibility for $\tau_g$ not to be a good witness, would be that there is exactly one word $x$ that visits $\tau_g$ first and exactly one word $u$ that recurs on $\tau_g$. Then all words $v$ with $\tpN(v) = \tau_g$ are of the form $xu^j$ for $j \ge 0$. Since there are infinitely many roots in $\tau_g$, we get that infinitely many of the $xu^j$ are roots. Hence, $x$ and $u$ cannot have the same root. So we are in case (2c), and $\tau_g$ is indeed a good witness.
\end{proof}

We close by showing that, if there is a good witness, then $\Ter(P)$ has infinitely many roots.
\begin{lemma}
    If there is a good witness $\tau_g$, then
    $\Ter{P}$ has infinitely many roots.
\end{lemma}

\begin{proof}
    Let $\tau_g$ be good witness.

    We first consider case (1), where we have infinitely many words that visit $\tau_g$ first.
    As they visit $\tau_g$ first, neither is a prefix of the other, and they must therefore have pairwise different roots.
    As a root of a word with a terminal transition profile has a terminal transition profile, this provides us with infinitely many roots in $\Ter{P}$.

    We continue with case (2) that $x$ visits $\tau_g$ first, $u$ recurs on $\tau_g$, and they do not have a joint root, so that $x^\omega \neq u^\omega$ holds.

    We now set $x_1=x$ and inductively construct:
    \begin{itemize}
        \item $k_j$ as the first position where ${x_j}^\omega$ differs from $u^\omega$,
        \item $\ell_k = \big\lceil k_j/|u| \big\rceil$ as a number chosen, such that $\big|u^{\ell_j}\big| \geq k_j$, and
        \item $x_{j+1} = x_j\cdot u^{\ell_j}$.
    \end{itemize}
    To see that ${x_j}^\omega$ differs from $u^\omega$, we note that, for all $j\geq 1$, $x_j = x \cdot u^i$ holds for some $i\geq 0$, so that ${x_j}^\omega=u^\omega$ would imply $x^\omega = u^\omega$ (contradiction).

    Note that $x_{j+1}$ is defined such that ${x_j}^\omega$ differs from ${x_{j+1}}^\omega$ at position $|x_j|+k_j$, and thus within the first iteration of $x_{j+1}$.

    We now have that all $x_i$ have the good transition profile $\tau_g$, and that they have pairwise different roots.
    Using again that the root of a word with a terminal transition profile has a terminal transition profile, this provides us with infinitely many roots in $\Ter{P}$.
\end{proof}

\rstIsItGood*

\begin{proof}
    Let $\N$ have $n$ states.

    We first, check whether $\tau_g$ is a terminal.
    To this end, we first need to be able check whether or not a transition profile $\tau$ is "accepting". But this is the case if, and only if, $\tau^i$ maps an initial state to a final state for some $i$ (which is indeed the case if, and only if, it holds for some $i \leq n$). Testing this can be done in time polynomial in $n$ (and in NL).

    Being able to check whether or not a transition profile is "accepting", we can guess a good transition profile $\tau_g$.

    For $\tau_g$, we then build the NFA $\N'$ with the same states, initial states, and final states as $\N$ over a one letter alphabet $\Sigma'=\{u\}$, where $\delta_{\N'}(q,u) \mapsto \tau_g(q)$.
    We now build the transition system $\T_{\N'}$ of size exponential in $n$, whose state space is the transition profiles of $\N'$, and where $\T_{\N'}(x)$ is the transition profile that maps each letter $q$ to $\delta(q,x)$.

    After reading $u^i$, $\T_{\N'}$ is in the state that represents the $i$-th power of $\tau_g$.

    Thus, we can check in NL in the size of $\T_{\N'}$, and thus is SPACE in polynomial in $n$, whether $\T_{\N'}$ can reach a "rejecting" transition profile.

    For the properties of the four cases (1) through (2c), we
    trace a transition system $\T_{\N}$ of size exponential in $\N$, whose state space is again the transition profiles of $\N$, but this time defined based on $\N$: $\T_{\N}(x)$ is the transition profile that maps each letter $q$ to $\delta(q,x)$.

    While $\T_{\N}$ is exponential in the size of $\N$, its states (the transition profiles of $\N$) and transitions are succinctly represented by $\N$.

    Checking cases (1) through (2c) can again be done by variations of reachability problems in $\T_{\N}$, and all of them are in NL.

    Recall that we have guessed $\tau_g$ before checking that it is terminating.

    It is straight forward to check in space logarithmic in $\T_{\N}$ (and thus polynomial in $n$) whether $\tau_g$ is reachable first by some word.

    Likewise, we can check if it is reachable first by two different words. For this, we simply traverse two copies concurrently for two potentially different words. We reject immediately if they reach $\tau_g$ before they have seen a different letter and otherwise accept when $\tau_g$ has been reached by both copies (not necessarily at the same time).

    Finally, we can check if $\tau_g$ is reachable in $\T_{\N}$ while reaching some transition profile $\tau$ (which we can guess) at least twice before reaching $\tau_g$.
    If this is the case, then there are infinitely many words; if not, then there cannot be a word that reach $\tau_g$ first, which is longer than the number of transition profiles (as all longer words must visit some transition profile twice), which means that there are only finitely many.

    To check if at least one (resp.\ two) words recur, we can use the same procedure as for checking if at least one (resp.\ two) words visit $\tau_g$ first, but changing the initial state from the identity to $\tau_g$.

    The results of these tests will tell us if we are in case (1), (2a), or (2b). If neither is the case, they will also tell us if there is exactly one word $x$ that visits $\tau_g$ first and exactly one word $u$ that recurs on $\tau_g$.

    If $x$ and $u$ share a root $r$, then we can guess in \PSPACE a transition profile $\tau$ such that (a) $\tau_g$ is equal to two different powers of $\tau$, and (b) that there is a word $w$ with $\T_{\N}(w) = \tau$.

    (b) is a standard reachability problem for $\T_{\N}$ (and thus again in space polynomial in $n$), while (a) can be done in a transition system $\T_{\N''}$ over the one letter alphabet, where $\N''$ is built from $\tau$ like $\N'$ was built from $\tau'$.
    For $\T_{\N''}$ we then check if $\tau_g$ is visited twice.

    If both is the case, then each such word $w$ is root of both $u$ and $x$, and vice versa, if $u$ and $x$ have a joint root $r$, then $\tau=\T_{\N}(r)$ satisfies both (a) and (b).
    Thus, we can also check if we are in case (2c).
\end{proof}%
  \section{Details on Succinctness Results from~\cref{section:comparison}}
\label{section:appendixcomparison}
\lowerboundCombinedToSatFDFA*
\begin{proof}
  We fix the alphabets $\Sigma = \{\sigma, \tau, \gamma\}$ and $\Sigma_\# = \Sigma \cup \{\#\}$.
  For a fixed $n$, each letter $a \in \Sigma$ corresponds to a function $f_a : \{1,\dotsc,n\} \to \{1,\dotsc,n\}$.
  Specifically, we set
  \begin{align*}
    f_\sigma(i) =
    \begin{cases}
      i + 1 & \text{if } i < n \\
      1     & \text{otherwise}
    \end{cases}
    \quad
    f_\tau(i) =
    \begin{cases}
      3 - i & \text{if } i \in \{1,2\} \\
      i     & \text{otherwise}
    \end{cases}
    \quad
    f_\gamma(i) =
    \begin{cases}
      1 & \text{if } i \in \{1,2\} \\
      i & \text{otherwise}
    \end{cases}
  \end{align*}
  which means $\sigma$ corresponds to a rotation of the numbers, $\tau$ transposes the first and second numbers, and $\gamma$ maps the first two numbers to the same number.
  We extend the definition of $f$ to words through function composition, i.e.~$f_{a_1 \dotsc a_k} = f_{a_1} \circ \dotsc \circ f_{a_k}$.
  
  The language family we now choose ensures that the progress DFA of any saturated FDFA for $L_n$ must contain a state for each mapping $f: \{1,\dotsc,n\} \to \{1,\dotsc,n\}$.
  At the same time, we would like $L_n$ to be recognized by a saturated FDWA with $(1, n+2)$ states.
  For this, we use
  \begin{align*}
    L_n = \{w \in \Sigma_\#^\omega \mid \text{$w$ contains infinitely many infixes $\#u\#$ such that $f_u(1) = 1$}\}
  \end{align*}
  consisting of all words where infinitely often between two $\#$ a function is encoded that sends $1$ to $1$.
  Clearly, $L_n$ is a prefix-independent language, so the leading TS has a state with self-loops over each letter in $\Sigma_{\#}$.
  We first show that it is possible to recognize $L_n$ through a "saturated" "FDWA" with $(1, n+2)$ states.
  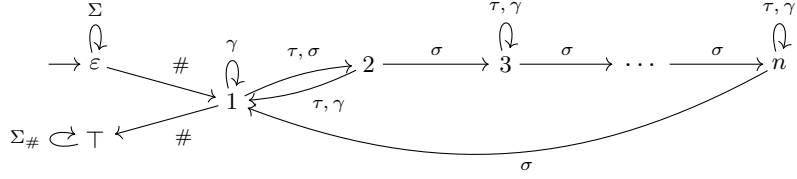
\begin{figure}[ht]
    \centering
    \begin{tikzpicture}[automaton, node distance=18mm]
      \node[state, initial] (e) {$\eps$};
      \node[state, right=of e, yshift=-5mm] (1) {$1$};
      \node[state, right=of 1, yshift=5mm] (2) {$2$};
      \node[state, right=of 2] (3) {$3$};
      \node[right=of 3] (dots) {$\dotsc$};
      \node[state, right=of dots] (n) {$n$};
      \node[state, left=of 1, yshift=-5mm,accepting] (accsink) {$\top$};

      \draw[->]
      (e)
      edge[loop above] node {$\Sigma$} (e)
      edge node {$\#$} (1)
      (1)
      edge[loop above] node {$\gamma$} (1)
      edge[bend left=10] node[pos=0.8] {$\tau,\sigma$} (2)
      edge node {$\#$} (accsink)
      (2)
      edge[bend left=10] node {$\tau, \gamma$} (1)
      edge node {$\sigma$} (3)
      (3)
      edge[loop above] node {$\tau, \gamma$} (3)
      edge node {$\sigma$} (dots)
      (dots)
      edge node {$\sigma$} (n)
      (n)
      edge[loop above] node {$\tau, \gamma$} (n)
      edge[bend left=25] node {$\sigma$} (1)
      (accsink)
      edge[loop left] node {$\Sigma_\#$} (accsink);
    \end{tikzpicture}
    \caption{
      The progress automaton of FDWA $\P_n$ constructed in the proof
      of~\cref{lemma:lowerbound:combinedToSatFDFA}.
      It accepts the language $L_n$ with $n+2$ states by keeping
      track of the number that $1$ is mapped to.
      The state $\top$ is an accepting sink.
    }
    \label{figure:lowerboundfwpm}
  \end{figure}
  Intuitively, the progress automaton of FDWA $\W_n$ depicted in~\cref{figure:lowerboundfwpm} recognizes $L_n$ by keeping track the number that $1$ is mapped to.

  The progress DFA of a saturated FDFA for $L_n$ on the other hand has to memorize the specific mapping encoded by the sequence of $\Sigma$-symbols since the last $\#$.
  To see this, assume that $\F = (\T, \D)$ is a "saturated FDFA" for $L_n$ with strictly less than $n^n$ states.
  This means that there exist two functions $f, g : \{1,\dotsc,n\} \to \{1,\dotsc,n\}$ which are not equal, but some words $x_f, x_g \in \Sigma^*$ encoding them reach the same state in $\D$.
  Let $k \leq n$ be the least integer satisfying $f(k) \neq g(k)$ which is guaranteed to exist since $f \neq g$.

  We set $i = f(k)$, $j = g(k)$ and consider the words $w_f = x_f \sigma^{n+1-i}\# \sigma^{k-1}$ and $w_g = x_g \sigma^{n+1-i}\# \sigma^{k-1}$.
  Since $f_{\sigma^{k-1} x_f \sigma^{n+1-i}}(1) = 1$, $\D$ must accept $\sigma^{k-1} x_f \sigma^{n+1-i}\#$.
  Because $\F$ is "saturated" and thus by~\cref{app:lemma:fdfacharacterisation} "loopshift-stable", we know that $\D$ must also accept $w_f$.
  But as $x_f$ and $x_g$ reach the same state, it must be that $\D$ also accepts $w_g$.
  This contradicts the assumption that $\D$ is saturated because $w_g$ is a rotation of $\sigma^{k-1} x_g \sigma^{n+1-i}\#$ and the mapping coded by $\sigma^{k-1} x_g  \sigma^{n+1-i}$ sends $1$ to some number other than $1$ since $x_g$ codes a function with $g(k) = j$ and $i \neq j$.

  Overall, we see that for each function $f : \{1,\dotsc,n\} \to \{1,\dotsc,n\}$, the progress DFA $\D_\eps$ of a "saturated FDFA" for $L_n$ must contain at least one state.
  As there are $n^n$ such functions, the lower bound on the size of a saturated FDFA for $L_n$ follows.
\end{proof}

We now provide a proof for~\cref{lemma:incomparableSatFDWAalsatFDFA}, which establishes that "almost saturated FDFAs" and "saturated FDWAs" are, in general, incomparable in size.
For the purpose of readability, we include the full statement below.

\incomparableSatFDWAalsatFDFA*

In the following, we show two auxiliary statements that are then combined to establish the lemma.

\begin{lemma}
    Let $n > 0$ and consider the language $L_n$ of words over $\Sigma_n = 2^{\{1,\dots,2n\}} \setminus \{\emptyset\}$ such that a word $w$ belongs to $L_n$ if some $i \in \{1,\dots,2n\}$ appears only finitely often in the letters of $w$.
    $L_n$ is recognized by a "saturated FDWA" of "size@@family" $(1,2n+1)$, but every "almost saturated FDFA" for $L_n$ with a trivial leading transition system has a progress automaton with at least $2^{\frac{n}{2}}$ states.
    \label{helper:incomparableSatFDWAalsatFDFA:first}
\end{lemma}
\begin{proof}
  We begin by showing that $L_n$ can be recognized by a "saturated FDWA" of size $(1,2n+1)$.
  Clearly, $L_n$ is prefix-independent, so we can just use a trivial TS $\T$ in the FDWA. 
  Now we describe the progress automaton $\D$.
  For $1 \leq i \leq k$, we define
  \begin{align*}
    G_i & = \{A \subseteq \{1,\dotsc,2n\} \mid i \notin A\} \qquad B_i = \{A \subseteq \{1,\dots,2n\} \mid i \in A\}.
  \end{align*}
  The progress automaton of an FDWA for $L_n$ can use one state for every $i$, ordered in a chain as follows:

  
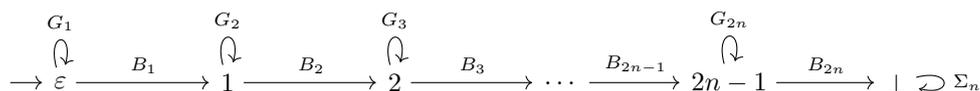
\begin{figure}[ht]
    \centering
    \begin{tikzpicture}[automaton, node distance=22mm]
      \node[initial] (e) {$\eps$};
      \node[right=of e] (1) {$1$};
      \node[right=of 1] (2) {$2$};
      \node[right=of 2] (dots) {$\dots$};
      \node[right=of dots] (2nminus1) {$2n-1$};
      \node[right=of 2nminus1] (bot) {$\bot$};

      \draw[->]
      (e)
      edge[loop above] node {$G_1 $} (e)
      edge node {$B_1 $} (1)
      (1)
      edge[loop above] node {$G_2 $} (1)
      edge node {$B_2 $} (2)
      (2)
      edge[loop above] node {$G_3$} (2)
      edge node {$B_3$} (dots)
      (dots)
      edge node {$B_{2n-1}$} (2nminus1)
      (2nminus1)
      edge[loop above] node {$G_{2n}$} (2nminus1)
      edge node {$B_{2n}$} (bot)
      (bot)
      edge[loop right] node {$\Sigma_n $} (bot);
    \end{tikzpicture}
  \end{figure}
  A periodic word $u^\omega$ will reach $\bot$ precisely if every number in $\{1,\dots,2n\}$ appears.
  Here, except the state $\bot$, every state is accepting.
  Since this property is naturally invariant under rotations of $u$, the given FDWA is "saturated".

  For the other direction, let $\F = (\T, \D)$ be a an "almost saturated" "FDFA" where $\T$ has only one state.
  To prove our lower bound, we first show an auxiliary claim.
  Consider two sets $X, Y \subseteq \{1,\dots,2n\}$ such that $X \not= Y$ and $|X| =|Y| = n$.
  We write $\bar{X}$ and $\bar{Y}$ denote their complement, i.e., $\{1,\dots,2n\}\setminus X$ and $\{1,\dots,2n\}\setminus Y$, respectively.
  We claim that if $X$ and $Y$ lead to the same state in $\D$, then $\bar{X}$ and $\bar{Y}$ do not lead to the same state.
  Assume to the contrary that $\D(X) = \D(Y) = q$ and $\D(\bar{X}) = \D(\bar{Y}) = p$.
  Since $(X\bar{X}u)^\omega$ and $(Y\bar{Y}u)^\omega$ are not in $L_n$ for every $u \in \Sigma_n^*$, we get that $\bar{X}$ and $\bar{Y}$ lead to a rejecting sink from $q$.
  Similarly, $X$ and $Y$ must lead to a rejecting sink from $p$.
  But then $\D$ reaches a rejecting sink on both $X\bar{Y}$ and $\bar{Y}X$.
  This is a contradiction because $|X| = |Y| = n$ and $X \not= Y$ implies $X \cup \bar{Y} \not= \{1, \ldots, 2n\}$, and therefore $(X\bar{Y})^\omega \in L_n$.

  The number of sets $X \subseteq \{1,\dotsc,2n\}$ of size $n$ is ${\binom{2n}{n}} \geq 2^n$.
  So in a deterministic transition system of size $k$, at least $\frac{2^n}{k}$ distinct sets of size $n$ will reach the same state $q$.  By applying the auxiliary claim, all their complements need to reach pairwise different states, which implies the following inequalities on the number of states
  \[
    k \geq \frac{2^n}{k} \quad \iff \quad k^2 \geq 2^n \quad \iff \quad k \geq 2^{\frac n 2}
  \]
  giving us the claimed lower bound.
\end{proof}

\begin{lemma}
    Let $n > 0$ and consider the language $L'_n$ over the alphabet $\{0,1\}$ consisting of all words $w$ with infinitely many occurrences of an infix $0u0$ for some $|u| \leq n$.
    $L'_n$ is recognized by an "almost saturated FDFA" of "size@@family" $(1,n+3)$, and by an "FDWA" of the same size.
    Every "saturated FDWA" with trivial leading transition system for the complement of $L_n'$ has a progress automaton with at least $2^n$ states, and every "almost saturated FDFA" for $L_n'$ with trivial leading transition system has a progress DFA with at least $\frac{2^n}{2n}$ states.
    \label{helper:incomparableSatFDWAalsatFDFA:second}
\end{lemma}
\begin{proof}
  Consider the family of languages where infinitely often two zeroes appear precisely $n$ positions apart, which is formally defined as
  \begin{align*}
    L'_n & = \{ w\in \Sigma^\omega \mid \text{ $w$ has infinitely many infixes $0u0$ for $u \in \Sigma^{n-1}$}\}.
  \end{align*}
  An "almost saturated FDFA" $\F'$ for $L_n'$ can use a trivial "leading transition system" and ``guess'' the start of an infix $0u0$ for some $|u| = n-1$ on the looping part.
  It needs a "progress DFA" with $n+3$ to verify that the zeroes indeed are $n-1$ symbols apart.
  We illustrate this progress automaton in~\cref{fig:napartillustration}.
  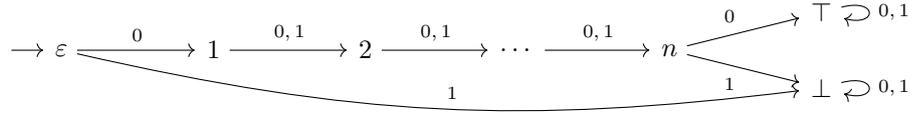
\begin{figure}[ht]
    \centering
    \begin{tikzpicture}[automaton, node distance=20mm]
      \node[initial] (e) {$\eps$};
      \node[right=of e] (1) {$1$};
      \node[right=of 1] (2) {$2$};
      \node[right=of 2] (dots) {$\dots$};
      \node[right=of dots] (n) {$n$};
      \node[right=of n, yshift=5mm] (top) {$\top$};
      \node[right=of n, yshift=-5mm] (bot) {$\bot$};

      \draw[->]
      (e)
      edge node {$0$} (1)
      edge[bend right=10] node {$1$} (bot)
      (1)
      edge node {$0,1$} (2)
      (2)
      edge node {$0,1$} (dots)
      (dots)
      edge node {$0,1$} (n)
      (n)
      edge node {$0$} (top)
      edge node[swap] {$1$} (bot)
      (bot)
      edge[loop right] node {$0,1$} (bot)
      (top)
      edge[loop right] node {$0,1$} (top)
      ;
    \end{tikzpicture}
    \caption{Illustration of the progress automaton (the "leading TS" is trivial) for the "almost saturated FDFA"/"FDWA" constructed in~\cref{helper:incomparableSatFDWAalsatFDFA:second}.}
    \label{fig:napartillustration}
  \end{figure}
  When viewed as an "FDWA", $\F'$ accepts the same language.

  \newcommand{\Lc}{\ensuremath{\overline{L_n'}}}
  For the second claim, we consider the complement of $L_n'$, which we denote $\Lc$.
  A word belongs to $\Lc$ if it contains \emph{finitely} many infixes $0u0$ for some $|u| = n-1$.
  
  Assume that $\W = (\T, \B)$ is a "saturated FDWA" with a trivial leading transition system, and that $\W$ recognizes $\Lc$.
  We show that every word $u \in \Sigma^n$ must lead to a different state in $\B$, which gives the desired lower bound on the size as $|\Sigma^n| = 2^n$.
  For a word $u = a_1\dotsc a_n$, we use $\tilde{u}$ to denote the word $(1-a_1)(1-a_2)\dots(1-a_n)$ where the bit in every position is flipped.
  Consider the sequence $q_0 = \iota, q_{i+1} = \delta^*(q_i, u\tilde{u})$, i.e. $q_i$ is the state reached on $(u\tilde{u})^i$.
  Then there must be $i < j \leq n$ such that $q_i = q_j$.
  We use $q_u$ to refer to the first state that is repeated at some later position in the sequence, meaning $q_u$ is reached by infinitely many powers of $u\tilde{u}$.

  Assume that for $u \neq v$ of length $n$ we have that $q_u = q_v =: p$.
  Pick $r > 0$ to be minimal with the property that $(u\tilde{u})^r$ reaches $q_u$, and let $k,l > 0$ be such that both $(u\tilde{u})^k$ and $(v\tilde{v})^l$ loop on $p$, that is, $\delta^*(p,(u\tilde{u})^k) = \delta^*(p,(v\tilde{v})^l) = p$.
  Let $m$ be such that $r+m$ is a multiple of $k$ and consider the word $x = (u\tilde{u})^r(v\tilde{v})^l(u\tilde{u})^m$.
  Note, that $x^\omega$ can be written as $(u\tilde{u})^k\left((v\tilde{v})^l(u\tilde{u})^{m+r}\right)^\omega$.
  Thus, $x^\omega$ reaches $q_u$ and takes the cycles around $q_u$ induced by $(u\tilde{u})^k$ and $(v\tilde{v})^l$ infinitely often.
  Therefore, since both $(u\tilde{u})^\omega$ and $(v\tilde{v})^\omega$ accepted by $\B$, also $x^\omega$ is accepted by $\B$.

  Clearly, $x^\omega$ contains infinitely many infixes $\tilde{u}v$ and $\tilde{v}u$.
  Because $u \neq v$, there exists a position $i \leq n$ such that $u_i \neq v_i = 1 - u_i$.
  By definition, we have $\tilde{u}_i = v_i$ and $\tilde{v}_i = u_i$.
  Now if $\tilde{u}_i = 0$, then $\tilde{u}v$ contains an infix $0x0$ for $|x| = n$ starting at position $i$.
  Otherwise, $\tilde{v}_i = u_i = 0$ and we find such an infix in $\tilde{v}u$.
  Hence, $x^\omega \notin \Lc$ but $\B$ accepts $x^\omega$, which is a contradiction.
  So overall, we have shown that in $\B$, there must exist at least $|\Sigma^n| = 2^n$ distinct states and the claim is established.

  We can use a similar proof technique to show an exponential lower bound on the size of an "almost saturated FDFA".
  Note, however, that in this case not every $(u\tilde{u})^\omega$ for $u \in \Sigma^*$ has to be accepted.
  Indeed, an almost saturated FDFA may only accept the $\omega$-power of one of the $2n$ rotations of $u\tilde{u}$.
  Overall, this gives a slightly lower but still exponential bound on the number of distinct states.
\end{proof}

We can now combine the preceding lemmas to obtain the proof of~\cref{lemma:incomparableSatFDWAalsatFDFA}.
The first claim is established by~\cref{helper:incomparableSatFDWAalsatFDFA:first}.
The third claim immediately follows from~\cref{helper:incomparableSatFDWAalsatFDFA:second}.
For the second claim, we observe that saturated FDWA can be complemented in place by swapping accepting and rejecting states.
This implies that a saturated FDWA for $L_n'$ with trivial "leading TS" also needs a progress automaton of size $2^n$, and the claim is established.

In the following, we give a full proof of~\cref{lemma:expgapsatFDFAsynFDFAfullysatFDFA}.
The families of languages that we use for both claimed bounds, make use of the alphabet $\Sigma^\to$ from~\cref{lemma:lowerbound:combinedToSatFDFA} for encoding mappings from $\{1,\dots,n\}$ to $\{1,\dots,n\}$.
For a word $u_i$, we write $m_{u_i}$ for the mapping that it defines.

\expgapsatFDFAsynFDFAfullysatFDFA*
\begin{proof}
For the first claim, we consider words of the form $u_0\#u_1\#\dots$ over $\Sigma^\to\dot\cup\{\#\}$ where $\#$ is a fresh symbol.
Such a word is in $L_n$, if $1$ is a fixpoint for infinitely many $u_i$, which means $m_{u_i}(1) = 1$ for infinitely many $i$.
We obtain a small "saturated FDFA" $(\T, \D)$ of "size@@family" $(n,2n)$ for $L_n$ as follows.
In the leading transition system $\T$, we track the number that $1$ is sent to under the mapping encoded by the word read since the last $\#$.
This is clearly possible with $n$ states, which we refer to as $1,\dots, n$, and leads to $n$ progress automata.
The progress automaton $\D_i$ accepts precisely the words $x\#y$ such that $m_x(i) = 1$ and $m_y(1) = i$.
Clearly, this can be done on the disjoint union of two copies of $\T$, by first ensuring that $i$ is sent to $1$, switching to a second copy on $\#$, and subsequently checking that $1$ is mapped to $i$.

The "syntactic FDFA" on the other hand has a trivial leading transition system, since membership in $L_n$ depends only on the infinite suffix.
The syntactic FDFA, the progress automaton for $u$ uses the syntactic progress congruence $\approx_u^L$ as transition system, where $x\approx_u^L y$ if for all $z$ with $uxz\sim_L u$ holds that $u(xz)^\omega \in L$ iff $u(yz)^\omega \in L$.
Now consider two words $x, y$ such that $m_x \neq m_y$.
We can find a word $w$ such that $1$ is a fixpoint of $m_w \circ m_x$ but not of $m_w \circ m_y$.
This shows that $x$ and $y$ must lead to different states in the progress automaton of the syntactic FDFA and since there are $n^n$ mappings from $\{1,\dots,n\}$ to itself, the bound is established.

Towards the second claim, the idea is to modify $L_n$ in a way that blows up the leading transition system, which makes it easier for the individual progress automata.
We introduce an additional letter $\sharp$ and separate words to force the leading transition system to keep track of which symbol was last seen.
To do this, we define $L_n'$ to be the union of $\L_n$ with the set of all words that contain an infix $\#u_ix$ where $x$ contains no $\#$ and $m_i(1)$ occurrences of $\sharp$.
This means that now two words ending in suffixes $x, y \in (\Sigma^\to)^*$ with $m_x(1) < m_y(1)$ can be separated by $\sharp^{m_x(i)}$.
Thus, the leading transition system of the "syntactic FDFA" for $L_n'$ has $n+1$ states, one for the numbers $1,\dots,n$, and a sink state.
The progress DFA for the sink is universal, while the one for $i$ behaves just as the progress DFA for $i$ in the saturated FDFA from the first claim does.
Thus, the size of the syntactic FDFA for $L_n'$ is $(n+1, 2n)$.

In contrast to this, a fully saturated FDFA for $L'$ cannot simply ignore words if they do not loop on a state, and it has to correctly classify all possible loops.
By applying a similar argument to before, we can show that any two words $x,y$ encoding mappings that differ in the image of $1$ must be separated in the progress automaton for $\eps$ in a "fully saturated FDFA" for $L_n'$.
Thus, in any such FDFA for $L_n'$ there must be a progress DFA with at least $n^n$ states.
\end{proof}%
  
\section{FDWA and FDFA with duo-normalized acceptance are the same}\label{section:appendix-duo-normal}
As mentioned in the related work of \Cref{section:introduction}, we can show that FDFAs with duo-normalized acceptance and FDWAs are basically the same model by showing that they can be translated into each other by just rewriting the acceptance condition and keeping the transition structure.

Let $\F = (\T, \D)$ be an FDFA.
With a normalized acceptance condition, the FDFA $\F$ only cares about whether normalized representations $(u, x) \in \Sigma^*\times \Sigma^+$ such that $\T(u) = \T(ux)$ are accepted or not.
That is, an FDFA with a normalized acceptance condition accepts $(u, x)$ if $(u, x)$ is a normalized representation and $x \in \LangDFA{\D_u}$.
Further, a normalized representation $(u, x)$ is said to be \emph{duo-normalized} if $\D_u(x) = \D_u(x\cdot x)$~\cite{FismanGZ24}.
Similarly, an FDFA with a duo-normalized acceptance condition accepts $(u, x)$ if $(u, x)$ is a duo-normalized representation and $x \in \LangDFA{\D_u}$. And it is saturated if for each ultimately periodic word, it accepts all duo-normalized representations of that word or none. 

The translation from FDWA to FDFA with duo-normalized acceptance is immediate.
\begin{lemma}
Each saturated FDWA for $L$ also defines $L$ and is saturated when interpreted as FDFA with duo-normalized acceptance.
\end{lemma}
\begin{proof}
Let $u \in \Sigma^*$ and $v \in \Sigma^+$ such that $v$ loops on the leading state reached by $u$, and on the progress state $p$ reached by $v$. Then then by definition of the semantics of FDWA, $uv^\omega \in L$ iff  $p$ is accepting.
\end{proof}

For the translation into the other direction, we need the following lemma, which summarizes properties of what is called the precise family of weak priority mappings for a regular $\omega$-language $L$ from~\cite{BohnL24}. The proof that we provide for the lemma mainly explains where to find the corresponding properties in \cite{BohnL24}.

\begin{lemma}[\cite{BohnL24}]\label{lem:precise-mapping}
  For each regular $\omega$-language $L$, there is a number $k \in \mathbb{N}$ and a function $\pi: \Sigma^* \times \Sigma^+ \rightarrow \{0,\ldots,k\}$ such that
  \begin{enumerate}
  \item For all $u,u',\in \Sigma^*$ with $u \sim_L u'$, and all $v \in \Sigma^+$:\; $\pi(u,v) = \pi(u',v)$.
  \item For all $u,x,v,y \in \Sigma^*$, and all $v \in \Sigma^+$:\; $\pi(u,xvy) \le \pi(ux,v)$
  \item There is $n \in \mathbb{N}$ such that for all $u \in \Sigma^*, v \in \Sigma^+$ and $m \ge n$:\;
$
\pi(u,v^m) \text{ is even iff $uv^\omega \in L$.}
$
  \end{enumerate}
\end{lemma}
\begin{proof}
  In Definition~3.5 of \cite{BohnL24}, for $L$ and a right-congruence $\sim$ that refines $\sim_L$, a family of priority mappings $\pi_c^\sim:\Sigma^+ \rightarrow \{0,\ldots,k-1\}$ is defined for each class $c$ of $\sim$, where $k$ is minimal such that $L$ can be accepted by a deterministic parity automaton with priorities in $\{0, \ldots, k-1\}$. For $u \in \Sigma^*$ let us write $\pi_u^\sim$ for $\pi_c^\sim$ where $c$ is the $\sim_L$-class $c$ of $u$.
  
  We choose $\sim := \sim_L$ and define $\pi(u,v) := \pi_u^\sim(v)$. 

  Then the first property is immediate by definition. The second property is a combination of the properties \textit{weak} and \textit{monotonic}: Let $u,x,v,y \in \Sigma^*$ and $v \in \Sigma^+$. Each $\pi_c^\sim$ is weak (see comment before Definition~3.5 of \cite{BohnL24}), which means that $\pi_u^\sim(xvy) \le \pi_u^\sim(xv)$. And the family of the $\pi_x^\sim$ is monotonic (Definition 3.11 and Lemma 3.12 of \cite{BohnL24}), which means that $\pi_u^\sim(xv) \le \pi_{ux}^\sim(v)$. Together we obtain $\pi(u,xvy) \le \pi(ux,v)$.

  By Lemma 3.8 of \cite{BohnL24}, the family of the $\pi_c^\sim$ captures the periodic part of $L$, which means that the least priority that appears infinitely often among the $\pi_u^\sim(x)$ for $x$ a prefix of $v^\omega$ is even iff $uv^\omega \in L$. Since $\pi_u$ is weak, this means that there is some $n \in \mathbb{N}$ such that all $\pi_u(v^m)$ for $m \ge n$ are even if $uv^\omega \in L$, and all $\pi_u(v^m)$ are odd if $uv^\omega \not\in L$. That the $n$ can be chosen independently of $u,v$ follows from the fact that each $\pi_u$ can be computed by a finite state Mealy machine \cite[Theorem 3.6]{BohnL24}.
\end{proof}

\begin{lemma}
Consider a saturated FDFA with duo-normalized acceptance that defines $L$. Obtain an FDWA by making those states accepting that are in an SCC of an accepting state that is reachable by a duo-normalized pair. Then this FDWA is saturated and defines $L$.
\end{lemma}
\begin{proof}
  Let $u,v$ be such that $v$ loops on the leading state reached by $u$. We have to show that $v^\omega$ is accepted by the progress DWA of $u$ iff $uv^\omega \in L$. Let $p$ be a progress state that is reached on $v^j$ and loops on $v^j$ for some $j > 0$. Then $(u,v^j)$ is a duo-normalized pair for $uv^\omega$. If $uv^\omega \in L$, then $p$ is accepting in the FDFA and by definition also accepting in the FDWA, and hence $v^\omega$ is accepted by the progress DWA of $u$. 

  If $uv^\omega \not\in L$, then we need to show that the SCC of $p$ does not contain an accepting state $q$ such that $q$ is reachable by a duo-normalized pair $(w,x)$. Assume the contrary. For obtaining a contradiction, we choose such problematic examples with minimal $\pi$-values for $\pi$ as in \cref{lem:precise-mapping}.

  So let $(u,v_1),(u,v_2)$ be duo-normalized pairs, such that $uv_1^\omega \in L$, $uv_2^\omega \not\in L$, and the states $p_1,p_2$ reached by $v_1,v_2$, respectively, in the progress congruence  are in the same SCC.

  Further choose $(u,v_1),(u,v_2)$ with these properties such that $\pi(u,v_1) + \pi(u,v_2)$ is minimal (with $\pi$ from \cref{lem:precise-mapping}). Let $x,y \in \Sigma^+$ be such that $p_1 \xrightarrow{x} p_2$ and $p_2 \xrightarrow{y} p_1$.

  Since, by the third property of \cref{lem:precise-mapping}, $\pi(u,v_1)$ is even and $\pi(u,v_2)$ is odd, we get that $\pi(u,v_1) \not= \pi(u,v_2)$.

  Consider the case that $\pi(u,v_1) < \pi(u,v_2)$, and define $z: = v_1xv_2^ny$ with $n$ as in the third property of \cref{lem:precise-mapping}. Then $(u,z)$ is a duo-normalized pair that loops on $p_1$, and hence $uz^\omega \in L$. Further, by the second and first property of $\pi$ from \cref{lem:precise-mapping}:
  \[
\pi(u,z^n) = \pi(u,(v_1xv_2^ny)^n) \le \pi(uv_1x,v_2^n) = \pi(u,v_2^n).
\]
By the third property of \cref{lem:precise-mapping}, we get that $\pi(u,z^n)$ is even, and thus  $\pi(u,z^n) < \pi(u,v_2^n)$. This contradicts the choice of $(u,v_1),(u,v_2)$ with $\pi(u,v_1) + \pi(u,v_2)$ being minimal.

The case $\pi(u,v_1) < \pi(u,v_2)$ uses a symmetric argument with $z: = v_2yv_1^nx$.
\end{proof}

}%
\end{document}